\documentclass[journal]{IEEEtran}
\usepackage{cite}
\usepackage{graphicx}
\usepackage{subcaption}
\usepackage[cmex10]{amsmath}
\usepackage{algpseudocode}
\usepackage{algorithmicx} 
\usepackage{array}
\usepackage{url}
\usepackage{caption}
\usepackage{epsfig}
\usepackage{amsfonts,amssymb,multirow,bigstrut,booktabs,ctable,latexsym}

\usepackage[mathscr]{eucal}
\usepackage{verbatim}
\usepackage{amsthm}
\usepackage{algorithm}

\begin{document}

\newtheorem{definition}{Definition}
\newtheorem{lemma}{Lemma}
\newtheorem{theorem}{Theorem}
\newtheorem{example}{Example}
\newtheorem{proposition}{Proposition}
\newtheorem{remark}{Remark}
\newtheorem{assumption}{Assumption}
\newtheorem{corrolary}{Corrolary}
\newtheorem{property}{Property}
\newtheorem{ex}{EX}
\newtheorem{problem}{Problem}
\newcommand{\argmin}{\arg\!\min}
\newcommand{\argmax}{\arg\!\max}
\newcommand{\st}{\text{s.t.}}
\newcommand \dd[1]  { \,\textrm d{#1}  }

\newcounter{mytempeqncnt}

%
% paper title
% Titles are generally capitalized except for words such as a, an, and, as,
% at, but, by, for, in, nor, of, on, or, the, to and up, which are usually
% not capitalized unless they are the first or last word of the title.
% Linebreaks \\ can be used within to get better formatting as desired.
% Do not put math or special symbols in the title.
\title{Optimal Secure Control with Linear Temporal Logic Constraints}

% author names and affiliations
% use a multiple column layout for up to three different
% affiliations

\author{Luyao Niu,~\IEEEmembership{Student Member,~IEEE,} and Andrew Clark,~\IEEEmembership{Member,~IEEE}% 
	\thanks{L. Niu and A. Clark are with the Department of Electrical and Computer Engineering, Worcester Polytechnic Institute, Worcester, MA 01609 USA.
		{\tt\{lniu,aclark\}@wpi.edu}}
	\thanks{This work was supported by NSF grant CNS-1656981.}
}

% make the title area

\maketitle
%%%%%%%%%%%%%%%%%%%%%%%%%%%%%%%%%%%%%%%%%%%%%%%%%%%%%%%%%%%%%%%%%%%%%
\begin{abstract}
%%%%%%%%%%%%%%%%%%%%%%%%%%%%%%%%%%%%%%%%%%%%%%%%%%%%%%%%%%%%%%%%%%%%%
Prior work on automatic control synthesis for cyber-physical systems under logical constraints has primarily focused on environmental disturbances or modeling uncertainties, however, the impact of deliberate and malicious attacks has been less studied. In this paper, we consider a discrete-time dynamical system with a linear temporal logic (LTL) constraint in the presence of an adversary, which is modeled as a stochastic game. We assume that the adversary observes the control policy before choosing an attack strategy. We investigate two problems. In the first problem, we synthesize a robust control policy for the stochastic game that maximizes the probability of satisfying the LTL constraint. A value iteration based algorithm is proposed to compute the optimal control policy. In the second problem, we focus on a subclass of LTL constraints, which consist of an arbitrary LTL formula and an invariant constraint. We then investigate the problem of computing a control policy that minimizes the expected number of invariant constraint violations while maximizing the probability of satisfying the arbitrary LTL constraint. We characterize the optimality condition for the desired control policy. A policy iteration based algorithm is proposed to compute the control policy. We illustrate the proposed approaches using two numerical case studies.
\end{abstract}


\IEEEpeerreviewmaketitle

%%%%%%%%%%%%%%%%%%%%%%%%%%%%%%%%%%%%%%%%%%%%%%%%%%%%%%%%%%%%%%%
\section{Introduction}\label{sec:introduction}
%%%%%%%%%%%%%%%%%%%%%%%%%%%%%%%%%%%%%%%%%%%%%%%%%%%%%%%%%%%%%%%

Cyber-physical systems (CPS) are expected to perform increasingly complex tasks in applications including autonomous vehicles, teleoperated surgery, and advanced manufacturing. An emerging approach to designing such systems is to specify a desired behavior using formal methods, and then automatically synthesize a controller satisfying the given requirements~\cite{lahijanian2012temporal,sadigh2016safe,wongpiromsarn2009receding,karaman2009sampling,loizou2004automatic}.

%Cyber-physical systems (CPS)  are expected to perform increasingly complex tasks with autonomy due to their large potential capabilities in computing and communication. Due to the increasing demand on autonomous CPS such as robots and autonomous vehicles, formal methods for precisely specifying system properties, verifying system performance and automatically synthesizing a controller satisfying these requirements are desired .

Temporal logics such as linear temporal logic (LTL) and computation tree logic (CTL) are powerful tools to specify and verify system properties \cite{baier2008principles}. In particular, LTL, whose syntax and semantics have been well developed, is widely used to express system properties. Typical examples include liveness (e.g., ``always eventually A''), safety (e.g., ``always not A''), and priority (e.g., ``first A, then B''), as well as more complex tasks and behaviors~\cite{baier2008principles}. For systems operating in stochastic environments or imposed probabilistic requirements (e.g., ``reach A with probability 0.9''), probabilistic extensions have also been proposed such as safety and reachability games that capture worst-case system behaviors~\cite{kattenbelt2009verification}. 

In addition to modeling uncertainties and stochastic errors \cite{wolff2012robust,ding2014optimal,kattenbelt2009verification}, CPS will also be subject to malicious attacks, including denial-of-service and injection of false sensor measurements and control inputs~\cite{CIAreport,koscher2010experimental}. For instance, power outages have been reported due to the penetration of attackers in power systems \cite{CIAreport}. Attacks against cars and UAVs are also reported in \cite{koscher2010experimental,kerns2014unmanned}. Unlike stochastic errors/modeling uncertainties, intelligent adversaries are able to adapt their strategies to maximize impact against a given controller, and thus exhibit strategic behaviors. Moreover, controllers will have limited information regarding the objective and strategy of the adversary, making techniques such as randomized control strategies potentially effective in mitigating attacks. In this case, control strategies synthesized using existing approaches may be suboptimal in the presence of intelligent adversaries because they are designed for CPS under errors and uncertainties. However, automatic synthesis of control systems in adversarial scenarios has received limited research attention.

In this paper, we investigate two problems for a probabilistic autonomous system in the presence of an adversary who tampers with control inputs based on the current system state. We abstract the system as a stochastic game (SG), which is a generalization of Markov decision process (MDP). We assume a concurrent Stackelberg information structure, in which the adversary and controller take actions simultaneously.  Stackelberg games are popular models in security domain \cite{paruchuri2008playing, zhu2011stackelberg,basilico2012patrolling,tambe2011security}. Turn-based Stackelberg games, in which a unique player takes action each time step, have been used to construct model checkers \cite{chen2013prism} and compute control strategy \cite{ding2013stochastic}, however, to the best of our knowledge, control synthesis in the concurrent Stackelberg setting has been less investigated. 

We focus on two problems. In the first problem, we are given an arbitrary LTL specification and focus on generating a control strategy such that the probability of satisfying the specification is maximized. In the second problem, we focus on a subclass of LTL specification that combine an arbitrary LTL specification with an invariant constraint using logical and connectives, where an invariant constraint requires the system to always satisfy some property. The specification of interest is commonly required for CPS, where the arbitrary LTL specification can be used to model properties such as liveness and the invariant property can be used to model safety property. We consider the scenario where the specification cannot be satisfied. Hence, we relax the specification by allowing violations on the invariant constraint and we select a control policy that minimizes the average rate at which invariant property violations occur while maximizing the probability of satisfying the LTL specification. We make the following specific contributions:

\begin{itemize}
	\item We formulate an SG to model the interaction between the CPS and adversary. The SG describes the system dynamics and the effects of the joint input determined by the controller and adversary. We propose a heuristic algorithm to compute the SG given the system dynamics.
	
	\item We investigate how to generate a control policy that maximizes the worst-case probability of satisfying an arbitrary specification modeled using LTL. We prove that this problem is equivalent to a zero-sum stochastic Stackelberg game, in which the controller chooses a policy to maximize the probability of reaching a desired set of states and the adversary chooses a policy to minimize that probability. We give an algorithm to compute the set of states that the system desires to reach. We then propose an iterative algorithm for constructing an optimal stationary policy. We prove that our approach converges to a Stackelberg equilibrium and characterize the convergence rate of the algorithm. %The goal of our problem is to select a controller strategy that maximizes the probability of satisfying the specifications. We prove that this problem is equivalent to a stochastic shortest path problem, and develop a polynomial-time algorithm for constructing the optimal strategies of the system and adversary.
	
	\item We formulate the problem of computing a stationary control policy that minimizes the rate at which invariant constraint violations occur under the constraint that an LTL specification must be satisfied with maximum probability. We prove that this problem is equivalent to a zero-sum Stackelberg game in which the controller selects a control policy that minimizes the average violation cost and the adversary selects a policy that maximizes such cost. We solve the problem by building up the connections with a generalized average cost per stage problem. We propose a novel algorithm to generate an optimal stationary control policy. We prove the optimality and convergence of the proposed algorithm.
	
	\item We evaluate the proposed approach using two numerical case studies in real world applications. We consider a remotely controlled UAV under deception attack given different LTL specifications. We compare the performance of our proposed approaches with the performance obtained using existing approaches without considering the adversary's presence. The results show that our proposed approach outperforms existing methods.  
	
\end{itemize}

The remainder of this paper is organized as follows. Section \ref{sec:related} presents related work. Section \ref{sec:preliminaries} gives background on LTL, SGs, and preliminary results on average cost per stage and average cost per cycle problem. Section \ref{sec: system model} introduces the system model. Section \ref{sec:formulation1} presents the problem formulation on maximizing the probability of satisfying a given LTL specification and the corresponding solution algorithm. Section \ref{sec:formulation2} presents the problem formulation and solution algorithm of the problem of minimizing the average cost incurred due to violating invariant property while maximizing the probability of satisfying an LTL specification. Two numerical case studies are presented in Section \ref{sec:simulation} to demonstrate our proposed approaches. Section \ref{sec: conclusion} concludes the paper.

\section{Related Work}\label{sec:related}

Temporal logics such as LTL and CTL are widely used to specify and verify system properties \cite{baier2008principles}, especially complex system behaviors. Multiple frameworks (e.g., receding horizon based \cite{wongpiromsarn2009receding}, sampling based \cite{karaman2009sampling}, sensor-based \cite{kress2007s,bhatia2010sampling,plaku2010motion}, probabilistic map based \cite{fu2016optimal}, multi-agent based \cite{loizou2004automatic}, and probabilistic satisfaction based \cite{lahijanian2010motion}) have been proposed for motion planning in robotics under temporal logic constraints. Control synthesis for deterministic system and probabilistic system under LTL formulas are studied in \cite{kloetzer2008fully} and \cite{ding2014optimal}, respectively. Switching control policy synthesis among a set of shared autonomy systems is studied in \cite{fu2016synthesis}. When temporal logic constraints cannot be fulfilled \cite{raman2011analyzing}, least-violating control synthesis problem is studied in \cite{tuumova2013minimum}. These existing works do not consider the impact of malicious attacks. 

Existing approaches of control synthesis under LTL constraints require a compact abstraction of CPS such as MDP \cite{ding2014optimal,sadigh2016safe,wolff2012robust}, which models the non-determinism and probabilistic behaviors of the systems, and enables us applying off-the-shelf model checking algorithms for temporal logic \cite{baier2008principles}. Robust control of MDP under uncertainties has been extensively studied \cite{nilim2005robust,wolff2012robust}. Synthesis of control and sensing strategies under incomplete information for turn-based deterministic game is studied in \cite{fu2016synthesisjointcontrol}. However, MDPs only model the uncertainties that arise due to environmental disturbances and modeling errors, and are only suitable for scenarios with a single controller.

For CPS operating in adversarial scenarios, there are two decision makers (the controller and adversary) and their decisions are normally coupled. Thus, MDP cannot model the system, and the robust control strategy obtained on MDP may be suboptimal to the CPS operated in adversarial environment. To better formulate the strategic interactions between the controller and adversary, SG is used to generalize MDP \cite{fudenberg1991game}. Turn-based two-player SG, in which a unique player takes action at each time step, have been used to construct model checkers \cite{chen2013prism} and abstraction-refinement framework for model checking \cite{kattenbelt2009verification, quatmann2016parameter, kattenbelt2010game}. Unlike the literature using turn-based games \cite{kattenbelt2009verification, quatmann2016parameter, kattenbelt2010game, chen2013prism}, however, we consider a different information structure denoted as concurrent SG, in which both players take actions simultaneously at each system state \cite{de2000concurrent}.

Several existing works using SGs focus on characterizing and computing Nash equilibria \cite{ma2011game,li2015jamming}, whereas in the present paper we consider a Stackelberg setting in which the adversary chooses an attack strategy based on the control policy selected by the system. The relationship between Nash and Stackelberg equilibria is investigated in \cite{korzhyk2011stackelberg}. A hybrid SG with asymmetric information structure is considered in \cite{ding2013stochastic}. The problem setting in \cite{ding2013stochastic} is similar to turn-based stochastic game, while the concurrent setting considered in this paper can potentially grants advantage to the controller (see \cite{paruchuri2008playing} for a simple example). Moreover, the problem setting in this paper leads to a more general class of control strategies. Specifically, mixed strategies are considered in this work. In particular, concurrent SGs played with mixed strategies generalizes models including Markov chains, MDPs, probabilistic turn-based games and deterministic concurrent games \cite{de2000concurrent}. The problem of maximizing the probability of satisfying a given specification consisting of safety and liveness constraints in the presence of adversary is considered in the preliminary conference version of this work \cite{Niu2018Secure}. Whereas only a restricted class of LTL specifications is considered in \cite{Niu2018Secure}, in this paper, we derive results for arbitrary LTL specifications and we also investigate the problem of minimizing the rate of violating invariant constraints.

CPS security is also investigated using game and control theoretic approaches. Secure state estimation is investigated in \cite{shoukry2016event,fawzi2014secure}. CPS security and privacy using game theoretic approach is surveyed in \cite{manshaei2013game}. Game theory based resilient control is considered in \cite{zhu2015game}. CPS security under Stackelberg setting and Nash setting are studied in \cite{zhu2011stackelberg} and \cite{zhu2011robust}, respectively. Stochastic Stackelberg security games have been studied in \cite{basilico2012patrolling,tambe2011security}. 

%%%%%%%%%%%%%%%%%%%%%%%%%%%%%%%%%%%%%%%%%%%%%%%%%%%%%%%%%%%%%%%
\section{Preliminaries}\label{sec:preliminaries}
%%%%%%%%%%%%%%%%%%%%%%%%%%%%%%%%%%%%%%%%%%%%%%%%%%%%%%%%%%%%%%%

In this section, we present background on LTL, stochastic games, and preliminary results on the average cost per stage (ACPS) and average cost per cycle (ACPC) problems. Throughout this paper, we assume that inequalities between vectors and matrices are component wise comparison.

\subsection{Linear Temporal Logic (LTL)}

An LTL formula consists of \cite{baier2008principles}
\begin{itemize}
	\item a set of atomic propositions $\Pi$;
	\item Boolean operators: negation ($\neg$), conjunction ($\land$) and disjunction ($\lor$).;
	\item temporal operators: next ($X$) and until ($\mathcal{U}$).
\end{itemize}
An LTL formula is defined inductively as
\begin{equation*}
\phi=True\mid\pi\mid\neg\phi\mid\phi_1\land\phi_2\mid X\phi\mid\phi_1~\mathcal{U}~\phi_2.
\end{equation*}
In other words, any atomic proposition $\phi$ is an LTL formula. Any formula formed by joining atomic propositions using Boolean or temporal connectives is an LTL formula. Other operators can be defined accordingly. In particular, implication ($\implies$) operator ($\phi\implies\psi$) can be described as $\neg\phi\lor\psi$; eventually $(\Diamond)$ operator $\Diamond\phi$ can be written as $\Diamond\phi=True~\mathcal{U}~\phi$; always ($\Box$) operator $\Box\phi$ can be represented as $\Box\phi=\neg \Diamond\neg\phi$.

The semantics of LTL formulas are defined over infinite words in $2^{\Pi}$ \cite{baier2008principles}. Informally speaking, $\phi$ is true if and only if $\phi$ is true at the current time step. $\psi~\mathcal{U}~\phi$ is true if and only if $\psi\land\neg\phi$ is true until $\phi$ becomes true at some future time step. $\Box\phi$ is true if and only if $\phi$ is true for the current time step and all the future time. $\Diamond\phi$ is true if $\phi$ is true at some future time. $X\phi$ is true if and only if $\phi$ is true in the next time step. A word $\eta$ satisfying an LTL formula $\phi$ is denoted as $\eta\models\phi$.

Given any LTL formula, a deterministic Rabin automaton (DRA) can be constructed to represent the formula. A DRA is defined as follows.
\begin{definition}
(Deterministic Rabin Automaton): A deterministic Rabin automaton (DRA) is a tuple $\mathcal{R}=(Q,\Sigma,\delta,q_0,\text{Acc})$, where $Q$ is a finite set of states, $\Sigma$ is a finite set of symbols called alphabet, $\delta:Q\times\Sigma\rightarrow Q$ is the transition function, $q_0$ is the initial state and $\text{Acc}=\{(L(1),K(1)),(L(2),K(2)),\cdots,(L(Z),K(Z))\}$ is a finite set of Rabin pairs such that $L(z),K(z)\subseteq Q$ for all $z=1,2,\cdots,Z$ with $Z$ being a positive integer. 
\end{definition}
A run $\rho$ of a DRA over a finite input word $\eta=\eta_0\eta_1\cdots\eta_n$ is a sequence of states $q_0q_1\cdots q_n$ such that $(q_{k-1},\eta_k,q_k)\in\delta$ for all $0\leq k\leq n$. A run $\rho$ is accepted if and only if there exists a pair $(L(z),K(z))$ such that $\rho$ intersects with $L(z)$ finitely many times and intersects with $K(z)$ infinitely often. Denote the satisfaction of a formula $\phi$ by a run $\rho$ as $\rho\models\phi$.

\subsection{Stochastic Games}
A stochastic game is defined as follows \cite{fudenberg1991game}.
\begin{definition}\label{def: CMDP}
(Stochastic Game): A stochastic game (SG) $\mathcal{SG}$ is a tuple $\mathcal{SG}= (S,U_C,U_A,Pr,s_0,\allowbreak \Pi,\mathcal{L})$, where $S$ is a finite set of states, $U_C$ is a finite set of actions of the controller, $U_A$ is a finite set of actions of an adversary, $Pr: S\times U_C\times U_A\times S\rightarrow [0, 1]$ is a transition function where $Pr(s, u_C, u_A, s^{\prime})$ is the probability of a transition from state $s$ to state $s^\prime$ when the controller's action is $u_C$ and the adversary's action is $u_A$.  $s_0\in S$ is the initial state. $\Pi$ is a set of atomic propositions. $\mathcal{L}:S\rightarrow 2^{\Pi}$ is a labeling function, which maps each state to a subset of propositions that are true at each state.
\end{definition}

Denote the set of admissible actions for the controller and adversary at state $s$ as $U_C(s)$ and $U_A(s)$, respectively. Given a finite set $S$, we use the Kleene star $S^\ast$ and the $\omega$ symbol $S^\omega$ to denote the set obtained by concatenating elements from $S$ finitely and infinitely many times, respectively. Given an SG, the set of finite paths, i.e, the set of finite sequence of states, can be represented as $S^\ast$, while the set of infinite paths, i.e., the set of infinite sequence of states, can be represented as $S^\omega$. The strategies (or policies) that players can commit to can be classified into the following two categories.
\begin{itemize}
\item \emph{Pure strategy}: A pure strategy gives the action of the player as a deterministic function of the state. Suppose the players commit to pure strategies. Then a pure control strategy is defined as $\mu:S^\ast\rightarrow U_C$, which gives a specific control action, and a pure adversary strategy is defined as $\tau:S^\ast\rightarrow U_A$.
\item \emph{Mixed strategy}: A mixed strategy determines a probability distribution over all admissible pure strategies. Suppose the players commit to mixed strategy. Then a control policy for the controller is defined as $\mu:S^*\times U_C\rightarrow [0,1]$, which maps a finite path and the admissible action to a probability distribution over the set of actions $U_C(s_k)$ available at state $s_k$. A policy $\tau$ for the adversary is defined as $\tau:S^*\times U_A\rightarrow[0,1]$.
\end{itemize}
In this paper we focus on  computing the optimal mixed strategy. When a specific action is assigned with probability one, then mixed strategy reduces to pure strategy. A control policy is \textit{stationary} if it is only a function of the current state, i.e., $\mu: S\times U_C\rightarrow [0,1]$ is only dependent on the last state of the path. A stationary policy is said to be \emph{proper} if the probability of satisfying the given specification after finite steps is positive under this policy. Given a pair of policies $\mu$ and $\tau$, an SG reduces to a Markov chain (MC) whose state set is $S$ and transition probability from state $s$ to $s^\prime$ is $P^{\mu\tau}(s,s^{\prime})\triangleq\sum_{u_C\in U_C(s)}\sum_{u_A\in U_A(s)}\mu(s,u_C)\tau(s,u_A)Pr(s,u_C,u_A,s^\prime)$. Given a path $\beta\in S^\omega$, a word is generated as $\eta_\beta=\mathcal{L}(s_0)\mathcal{L}(s_1)\cdots$. The probability of satisfying an LTL formula $\phi$ under policies $\mu$ and $\tau$ on $\mathcal{SG}$ is denoted as $Pr_{\mathcal{SG}}^{\mu\tau}=Pr\{\eta_\beta\models\phi:\beta\in S^\omega\}$.

In the following, we review a subclass of stochastic games, denoted as Stackelberg games, involving two players \cite{fudenberg1991game}. In the Stackelberg setting, player $1$ (also called leader) commits to a strategy first. Then player $2$ (also known as follower) observes the strategy of the leader and plays its best response. The information structure under Stackelberg setting can be classified into the following two categories.
\begin{itemize}
	\item \emph{Turn-based games}: Exactly one player is allowed to take action at each time step. Turn-based games are used to model asynchronous interaction between players.
	\item \emph{Concurrent games}: All the players take actions simultaneously at each time step. Concurrent games are used to model synchronous interaction between players.
\end{itemize}
Unlike \cite{kattenbelt2009verification, quatmann2016parameter, kattenbelt2010game, chen2013prism}, which are turn-based and played with pure strategies, in this paper, we focus on concurrent games with players committing to mixed strategies. To demonstrate the efficiency of mixed strategies in a concurrent game, we consider a robot  moving to the right in a 1D space. The action sets for the controller and adversary are $\{move, stay\}$. If the controller and action take the same action at a given time step, then the robot will follow the specified action, i.e., move one step to the right under the action pair $(move,move)$ and stay in the current location under action pair $(stay,stay)$. The goal of the robot is to move to the location immediately to the right of its starting location. When the controller commits to a pure strategy, say $move$, then the adversary will always take action $stay$, and the robot will remain at its starting location. On the other hand, if the controller plays a mixed strategy, e.g., choosing $move$ and $stay$ with equal probability $1/2$ at each time step, then the robot has a $1/2$ probability to reach the desired location at each time step, and hence will reach the desired location within finite time with probability 1. On the other hand, in a turn-based setting where the adversary observes the controller's action before choosing its action at each time step, the adversary will always be able to choose the opposite of the controller's action and prevent the robot from moving. Hence while mixed strategies are beneficial in the concurrent game formulation, they are not beneficial in the turn-based game for this case.

%To see the efficiency of committing to mixed strategy in concurrent games, we consider a robot under false data injection attack moving in a discrete grid world consists of two grids. The objective of the robot, who is the leader, is to move to its adjacent grid while the objective of the adversary, who is the follower, is to stop the transition of the robot. Suppose the action sets for the controller and the adversary are $\{stay,move\}$. When both players take same action, then the robot moves as indicated by the action with probability one, i.e., when $stay$ is taken by both players, the robot stay at its current grid and when $move$ is taken by both players, the robot move to its adjacent grid. When the actions taken by both players are different, then the robot stays at its current grid with probability one. Suppose the controller commits to a pure strategy, say `$move$', then the adversary will always play `$stay$' so that the robot will never move. When the controller commits to mixed strategy, however, then the controller can achieve its objective at some future time step. The optimal strategy of the controller is to assign uniform distributions over its action set, and it has $1/2$ probability to win the game at each round. Hence, the controller has probability one to win the game by committing to mixed strategy.

The concept of Stackelberg equilibrium is used to solve Stackelberg games. The Stackelberg equilibrium is defined formally in the following.
\begin{definition}\label{def: Stackelberg equilibrium}
	(Stackelberg Equilibrium): Denote the utility that the leader gains in a stochastic game $\mathcal{SG}$ under leader follower strategy pair $(\mu,\tau)$ and the utility that the follower gains as $\mathcal{Q}_L(\mu,\tau)$ and $\mathcal{Q}_F(\mu,\tau)$, respectively. A pair of leader follower strategy $(\mu,\tau)$ is a Stackelberg equilibrium if leader's strategy $\mu$ is optimal given that the follower observes its strategy and plays its best response, i.e., $\mu=\argmax_{\mu^\prime\in\boldsymbol{\mu}} \mathcal{Q}_{L}(\mu^\prime,\mathcal{BR}(\mu^\prime))$, where $\boldsymbol{\mu}$ is the set of all admissible policies of the controller and $\mathcal{BR}(\mu^\prime)=\{\tau:\tau=\argmax\mathcal{Q}_F(\mu^\prime,\tau)\}$ is the best response to leader's strategy $\mu^\prime$ played by the follower.
\end{definition}

\subsection{ACPS and ACPC Problems}

We present some preliminary results on the average cost per stage (ACPS) problem  and average cost per cycle (ACPC) problem on MDP without the presence of adversary in this subsection. Both problems focus on deterministic control policies $\mu:S\rightarrow U_C$ and unichain MDP. An MDP is said to be unichain if for any control policy $\mu$, the induced MC is irreducible, i.e, the probability of reaching any state from any state on the MC is positive. Denote the cost incurred at state $s$ by applying the deterministic control policy $\mu$ as $g(s,\mu(s))$. The transition probability from state $s$ to $s^\prime$ via action $u$ on MDP is denoted as $Pr(s,u,s^\prime)$. Each transition to a new state is viewed as a completion of a stage. Then the objective of ACPS problem is to minimize 
\begin{equation}\label{eq: preliminary ACPS}
J_{\mu}(s)=\limsup_{N\rightarrow\infty}\frac{1}{N} \mathbb{E}\left\{\sum_{k=0}^Ng(s_k,\mu(s_k))|s_0=s\right\}
\end{equation}
over all deterministic stationary control policies.

It has been shown that a gain-bias pair $(J_\mu,h_\mu)$ for ACPS problem satisfies the properties stated as follows.
\begin{proposition}[{\cite{bertsekas1995dynamic}}]\label{prop: ACPS}
Assume the MDP is unichain. Then:
\begin{itemize}
\item the optimal ACPS $J_{\mu}^*(s_0)$ associated with each control policy $\mu$ is independent of initial state $s_0$, i.e., there exists a constant $J_{\mu}^*$ such that $J_{\mu}^*(s_0)=J_{\mu}^*$ for all $s_0\in S$;
\item there exists a vector $h$ such that 
$J_{\mu}^*+h_\mu(s)=\min_{\mu}\Big\{g(s,\mu(s))
+\sum_{s^\prime\in S}Pr(s,\mu(s),s^\prime)h(s^\prime)\Big\}.$
\end{itemize}
\end{proposition}
%The pair $(J_\mu,h)$ is denoted as gain-bias pair.

We present some preliminary results on the ACPC problem on MDP in the following. Denote the set of states that satisfy LTL formula $\phi$ as $S_{\phi}$. A cycle is completed when $S_{\phi}$ is visited. Therefore, a path starting from $s_0$ and ending in $S_{\phi}$ completes the first cycle, and the path starting from $S_{\phi}$ after completing the first cycle completes the second cycle when coming back to $S_{\phi}$. Denote the number of cycles that have been completed until stage $N$ as $C(N)$. The ACPC problem is described as given an MDP and an LTL formula $\phi$, find a control policy $\mu$ that minimizes 
\begin{equation}\label{eq: preliminary ACPC}
J_{\mu}(s_0)=\underset{N\rightarrow\infty}{\lim\sup}~ \mathbb{E}\left\{\frac{\sum_{k=0}^Ng(s_k,\mu(s_k))}{C(N)}\Big|\eta_\mu\models\phi\right\},
\end{equation}
where $\eta_{\mu}=\mathcal{L}(s_0)\mathcal{L}(s_1)\cdots$ is the word generated by the path $s_0s_1\cdots$ induced by deterministic control policy $\mu$.

It has been shown that the following proposition holds for the ACPC problem.
\begin{proposition}[{{\cite{ding2014optimal}}}]
	Assume the MDP is unichain. Then:
	\begin{itemize}
		\item the optimal ACPC $J_{\mu}^*(s_0)$ associated with each control policy $\mu$ is independent of initial state $s_0$, i.e., there exists a constant $J_{\mu}^*$ such that $J_{\mu}^*(s_0)=J_{\mu}^*$ for all $s_0\in S$;
		\item there exists some vector $h$ such that the following equation holds
		$J_{\mu}^*+h(s)=\min_{\mu}\Big\{g(s,\mu(s))
		+\sum_{s^\prime\in S}Pr(s,\mu(s),s^\prime)h(s^\prime)
		+J_{\mu}^*\sum_{s^\prime\in S}Pr(s,\mu(s),s^\prime)\Big\}.$
	\end{itemize}
\end{proposition}
%\subsection{Dynamic Programming}
%
%In this paper, we use dynamic programming to calculate the optimal control policy. To demonstrate the optimality of our proposed approach, the following theorems are used for proof.
%\begin{theorem}\label{theo:contraction mapping}
%(\emph{Contraction Mapping Theorem}). Let $(X,d)$ be a metric space and $T:X\rightarrow X$. Then we call the map $T$ a contraction mapping on $X$ if there exists a $\beta\in[0,1)$ such that
%\begin{equation*}
%d(T(x),T(y))\leq\beta d(x,y)
%\end{equation*}
%for all $x$ and $y$ in $X$.
%\end{theorem}
%
%To see if a mapping is a contraction, we rely on the Blackwell's sufficient conditions presented in the following theorem.
%\begin{theorem}\label{theo:Blackwell's sufficient conditions}
%	(\emph{Blackwell's sufficient conditions for a contraction}\cite{granas2013fixed}). Let $X\subseteq\mathbb{R}^n$ and $B(x)$ be the space of bounded functions $f:X\rightarrow B(X)$. Then an operator $T: B(X)\rightarrow B(X)$ is a contraction with modulus $\rho$ if the following two conditions are satisfied:
%	\begin{itemize}
%		\item \emph{Monotonicity}: Denote the composition of $T$ with itself $k$ times as $T^{k}$. For any vectors $J,J^\prime\in\mathbb{R}^n$ satisfying $J\leq J^\prime$, we have $(T^{k}J)\leq (T^{k}J^\prime)$ for all $k=1,2,\cdots$.
%		\item \emph{Discounting}: There exists a constant $\rho\in(0,1)$ such that
%		\begin{equation*}
%		T(f+\theta)\leq Tf+\rho\theta
%		\end{equation*}
%		for all $f\in B(X)$, $x\in X$ and $\theta\geq0$.
%	\end{itemize}
%\end{theorem}

%%%%%%%%%%%%%%%%%%%%%%%%%%%%%%%%%%%%%%%%%%%%%%%%%%%%%%%%%%%%%%%
\section{System Model}\label{sec: system model}
%%%%%%%%%%%%%%%%%%%%%%%%%%%%%%%%%%%%%%%%%%%%%%%%%%%%%%%%%%%%%%%

In this section, we present the system model. We consider the following discrete-time finite state system
\begin{equation}\label{eq: dynamics}
x(t+1) = f(x(t),u_C(t),u_A(t),\vartheta(t)),~\forall t=0,1,\cdots,
\end{equation}
where $x(t)$ is the finite system state, $u_C(t)$ is the control input from the controller, $u_A(t)$ is the attack signal from the adversary, and $\vartheta(t)$ is stochastic disturbance.

In system \eqref{eq: dynamics}, there exists a strategic adversary that can tamper with the system transition. In particular, the controller and adversary jointly determine the state transition. For instance, an adversary that launches false data injection attack modifies the control input as $u(t)=u_C(t)+u_A(t)$; an adversary that launches denial-of-service attack manipulates the control input as $u(t)=u_C(t)\cdot u_A(t)$, where $u_A(t)\in\{0,1\}$.

In security domain, Stackelberg game is widely used to model systems in the presence of malicious attackers. In Stackelberg setting, the controller plays as the leader and the adversary plays as the follower. In this paper, we adopt the concurrent Stackelberg setting. The controller first commits to its control strategy. The adversary can stay outside for indefinitely long time to observe the strategy of the controller and then chooses its best response to the controller's strategy. However, at each time step, both players must take actions simultaneously. The system is given some specification that is modeled using LTL.

To abstract system \eqref{eq: dynamics} as a finite state/action SG, we propose a heuristic simulation based algorithm as shown in Algorithm \ref{alg: CMDP}, which is generalized from the approaches proposed in \cite{lahijanian2009probabilistic,wolff2012robust}. The difference between Algorithm \ref{alg: CMDP} and algorithms in \cite{lahijanian2009probabilistic,wolff2012robust} is that Algorithm \ref{alg: CMDP} considers the presence of adversary. Algorithm \ref{alg: CMDP} takes the dynamical system \eqref{eq: dynamics}, the set of sub-regions of state space $\{X_1, \cdots, X_n\}$ and actions as inputs. We observe that the choice of subregions $X_{1},\ldots,X_{n}$ may affect the accuracy of the model, however, choice of the subregions is beyond the scope of this work. %One approach to obtain the set of sub-regions is to use the atomic propositions for partitioning. %For instance, \cite{kloetzer2008fully} assumes the evaluation of each atomic proposition $\phi_i\in\Pi$ defines a hyperplane, i.e, $x\models\phi_i$ for all $x\in\mathbb{R}^n$ such that $c_i^Tx<d_i$ with $c_i\in\mathbb{R}^n$, $d_i\in\mathbb{R}$ and $x\models\neg\phi_i$ for all $x\in\mathbb{R}^n$ such that $c_i^Tx>d_i$ with $c_i\in\mathbb{R}^n$ and $d_i\in\mathbb{R}$. Therefore, given a set of propositions $\Pi$, the state space of system \eqref{eq: dynamics} is partitioned into a set of sub-regions $\{X_1,\cdots,X_N\}$ such that $\mathcal{L}(x_p)=\mathcal{L}(x_q)$ for all $x_p,x_q\in X_k$ with $1\leq k\leq N$. 
For each sub-region $X_{i}$ and pair of (control, adversary) inputs $(u_{C},u_{A})$, we randomly select $K$ sample states in $X_{i}$ and adversary and control inputs that map to $u_{C}$ and $u_{A}$. We compute the probability distribution over the set of sub-regions $\{X_{j}\}$ that the system can transition to following \eqref{eq: dynamics}, and update $Pr(X_{i},u_{C},u_{A},X_{j})$ accordingly for all $X_j$ (Algorithm \ref{alg: CMDP}). To approximate the transition probability, Monte Carlo simulation or particle filter can be used \cite{wolff2012robust, cappe2007overview, lahijanian2009probabilistic}.

\begin{center}
	\begin{algorithm}[!htp]
		\caption{Algorithm for constructing a stochastic game approximation of a system.}
		\label{alg: CMDP}
		\begin{algorithmic}[1]
			\Procedure{Create\_Stochastic\_Game}{ $X_{1},\ldots,X_{n}$}
			\State \textbf{Input:} Dynamics \eqref{eq: dynamics}, set of subsets $X_{1},\ldots,X_{n}$
			\State \textbf{Output:} Stochastic game  $\mathcal{SG}=(S,U_C,U_A,Pr,s_0,\Pi,\mathcal{L})$
			\State Initialize $K$
			\State $S=\{X_1,\ldots,X_n\}$ and $\mathcal{L}$ is determined accordingly
			\State Generate control primitive sets $U_C=\{u_{C_1},u_{C_2}\cdots,u_{C_\Xi}\}$ and $U_A=\{u_{A_1},u_{A_2}\cdots,u_{A_\Gamma}\}$
			\For{$i=1,\ldots,n$}
			\For{all $u_{C} \in U_{C}$ and $u_{A} \in U_{A}$}
			\For{$k=1,\ldots,K$}
			\State $x \leftarrow$ sampled state in $X_{i}$
			\State $\hat{u}_{C}, \hat{u}_{A} \leftarrow$ sampled  inputs from $u_{C},u_{A}$
			\State $j \leftarrow$ region containing $f(x, \hat{u}_{C}, \hat{u}_{A},\vartheta)$
			
			% \State $P(X_{i},u_{C},u_{A},X_{j}) \leftarrow P(X_{i},u_{C},u_{A}, X_{j}) + \frac{1}{K}$
			\State Invoke particle filter to approximate transition probabilities $Pr$ between sub-region $i$ and $j$ for all $i$ and $j$. 
			\EndFor
			\EndFor
			\EndFor
			\EndProcedure
		\end{algorithmic}
	\end{algorithm}
\end{center}

In the following, we present two applications in the security domain that can be formulated using the proposed framework.

\subsubsection{Infrastructure Protection in Power System}\label{ex: power system}
The proposed framework can capture attack-defense problems on power system as shown in the following. An attack-defense problem on power system is investigated in \cite{ma2011game}.

The players involved in this example are the power system administrator and adversary. The adversary aims to disrupt the transmission lines in power network, while the administrator deploys resources to protect critical infrastructures or repair damaged infrastructures. The dynamics of power system is modeled as $x(t+1)=f(x(t),u_C(t),u_A(t))$, where $x(t)$ is the state vector, $u_C(t)$ and $u_A(t)$ are the inputs from the administrator and adversary, respectively. Depending on the focus of the administrator, the state may contain bus voltages, bus power injections, network frequency, and so on. The actions of the administrator $U_C$ and adversary $U_A$, respectively, are the actions to protect (by deploying protection or repair resources) and damage (by opening the breakers at ends of) the transmission lines. If an attack is successful, then the transmission line is out of service, which will result in dramatic change on state vector. Thus the states evolve following the joint actions of administrator and adversary. Moreover, the probability of the occurrence of events, i.e., the transmission line is out of service, is jointly determined by the actions of adversary and administrator (or defender). The specifications that can be given to the system might include reachability (e.g., 'eventually satisfy optimal power flow equation': $\Diamond OPF$) and reactivity (e.g., 'if voltage exceeds some threshold, request load shedding from demand side': $\Box(voltage\_alarm\implies X~DR)$).

%Using the problem formulation discussed above, several typical attacks that have been investigated in CPS can be considered under our proposed framework including jamming attack, denial-of-service (DoS) attack and so on. We conclude this section by presenting two examples of attack models in CPS that can be modeled using our problem formulation.

\subsubsection{Networked Control System under Attacks}\label{ex: DoS attack}
In the following, we present an example on control synthesis for networked control system under deception attacks.

The system is modeled as a discrete linear time invariant system $x(k+1)=Ax(k)+Bu(k)+\vartheta(k),~k=0,1\cdots,$
where $x(k)$ is the system state, $u(k)$ is the compromised control input and $\vartheta(k)$ is independent Gaussian distributed disturbance. There exists an intelligent and strategic adversary that can compromise the control input of the system by launching deception attack. When the adversary launches deception attack on system actuator \cite{zhu2011stackelberg}, then the control input is represented as $u(k) = u_C(k)+u_A(k).$ Typical specifications that are assigned to the system include stability and safety (e.g., `eventually reach stable status while not reaching unsafe state': $\Diamond\Box\text{stable}\land\Box\neg\text{unsafe}$). 
%%%%%%%%%%%%%%%%%%%%%%%%%%%%%%%%%%%%%%%%%%%%%%%%%%%%%%%%%%%%%%%
\section{Problem Formulation - Maximizing Satisfaction Probability}\label{sec:formulation1}
%%%%%%%%%%%%%%%%%%%%%%%%%%%%%%%%%%%%%%%%%%%%%%%%%%%%%%%%%%%%%%%

In this section, we formulate the problem of maximizing the probability of satisfying a given LTL specification in the presence of an adversary. We first present the problem formulation, and then give a solution algorithm for computing the optimal control policy. 

%%%%%%%%%%%%%%%%%%%%%%%%%%%%%%%%%%%%%%%%%%%%%%%%%%%%%%%%%%%%%%%
\subsection{Problem Statement}
%%%%%%%%%%%%%%%%%%%%%%%%%%%%%%%%%%%%%%%%%%%%%%%%%%%%%%%%%%%%%%%

The problem formulation is as follows.
\begin{problem}\label{problem: max satisfaction probability}
Given a stochastic game $\mathcal{SG}$ and an LTL specification $\phi$, compute a control policy $\mu$ that maximizes the probability of satisfying the specification $\phi$ under any adversary policy $\tau$, i.e.,
\begin{equation}\label{eq: LTL maximization}
\max_{\mu}\min_{\tau}Pr_{\mathcal{SG}}^{\mu\tau}(\phi).
\end{equation}
\end{problem}

Denote the probability of satisfying specification $\phi$ as satisfaction probability. The policies $\mu$ and $\tau$ that achieve the max-min value of (\ref{eq: LTL maximization}) can be interpreted as an equilibrium defined in Definition \ref{def: Stackelberg equilibrium} in a zero-sum Stackelberg game between the controller and adversary, in which the controller first chooses a randomized policy $\mu$, and the adversary observes $\mu$ and selects a policy $\tau$ to minimize $Pr_{\mathcal{SG}}^{\mu\tau}(\phi|s)$. By Von Neumann's theorem \cite{neumann1928theorie}, the satisfaction probability at equilibrium must exist. We restrict our attention to the class of stationary policies, leaving the general case for future work. We have the following preliminary lemma.
\begin{lemma}\label{lemma:Stackelberg-Shapley}
Let satisfaction probability $v(s) = \max_{\mu}{\min_{\tau}{Pr_{\mathcal{SG}}^{\mu\tau}(\phi | s)}}$. Then
\begin{multline}\label{eq:Stackelberg-Shapley}
v(s) = \max_\mu\min_\tau\sum_{u_{C} \in U_{C}(s)}\sum_{u_{A} \in U_{A}(s)}\sum_{s^{\prime} \in S}\mu(s,u_{C})\\\tau(s,u_{A})v(s^{\prime}) Pr(s,u_{C},u_{A},s^{\prime}).
%\sum_{u_{C} \in U_{C}(s)}{u_{A} \in U_{A}(s)}{\sum_{s^{\prime} \in S}{\mu(u_{C})\tau(u_{A})v_{s^{\prime}}Pr(s, u_{C},u_{A},s^{\prime})}}.
\end{multline}
Conversely, if $v(s)$ satisfies (\ref{eq:Stackelberg-Shapley}), then $v(s) = \max_{\mu}{\min_{\tau}{Pr_{\mathcal{SG}}^{\mu\tau}(\phi |s)}}$. Moreover, the satisfaction probability $v$ is unique.
\end{lemma}
\begin{proof}
In the following, we will first show the forward direction. We let $n = |S|$ and define three operators $T_{\mu\tau}: [0,1]^{n} \rightarrow [0,1]^{n}$, $T_{\mu}: [0,1]^{n} \rightarrow [0,1]^{n}$, and $T: [0,1]^{n} \rightarrow [0,1]^{n}$.
\begin{IEEEeqnarray*}{rCl}
(T_{\mu\tau}v)(s) &=& \sum_{s^{\prime}}{Pr(s,\mu,\tau,s^{\prime})v(s^\prime)} \\
(T_{\mu}v)(s) &=& \min_{\tau}{\sum_{s^{\prime}}{Pr(s,\mu,\tau,s^{\prime})v(s^{\prime})}} \\
(Tv)(s) &=& \max_{\mu}{\min_{\tau}{\sum_{s^{\prime}}{Pr(s,\mu,\tau,s^{\prime})v(s^{\prime})}}}
\end{IEEEeqnarray*}
where $Pr(s,\mu,\tau,s^{\prime}) = \sum_{u_{C} \in U_{C}(s)}\sum_{u_{A} \in U_{A}(s)}\mu(s,u_{C})\allowbreak\tau(s,u_{A})\allowbreak Pr(s,u_{C},u_{A},s^{\prime})$. Suppose that $\mu$ is a Stackelberg equilibrium with $v(s)$ equal to the satisfaction probability for state $s$, and yet (\ref{eq:Stackelberg-Shapley}) does not hold. We have that $v = T_{\mu}v$, since $v$ is the optimal policy for the MDP defined by the policy $\mu$~\cite{bertsekas1995dynamic}. On the other hand, $T_{\mu}v \leq Tv$. Composing $T$ and $T_{\mu}$ $k$ times and taking the limit as $k$ tends to infinity yields $v = \lim_{k \rightarrow \infty}{T_{\mu}^{k}v} \leq \lim_{k \rightarrow \infty}{T^{k}v} \triangleq v^{\ast}.$ The convergence of $T^{k}v$ to a fixed point $v^{\ast}$ follows from the fact that $T$ is a bounded and monotone nondecreasing operator. Furthermore, choosing the policy $\mu(s)$ at each state as the maximizer of (\ref{eq:Stackelberg-Shapley}) yields a policy with satisfaction probability $v^{\ast}$. Hence $v \leq v^{\ast}$. If $v(s) = v^{\ast}(s)$ for all states $s$, then (\ref{eq:Stackelberg-Shapley}) is satisfied, contradicting the assumption that the equation does not hold. On the other hand, if $v(s) < v^{\ast}(s)$ for some state $s$, then $\mu$ is not a Stackelberg equilibrium.

We next show that the vector $v$ satisfying (\ref{eq:Stackelberg-Shapley}) is unique. Since every Stackelberg equilibrium satisfies (\ref{eq:Stackelberg-Shapley}), if the vector $v$ is unique, then the vector $v$ must be a Stackelberg equilibrium. Suppose that uniqueness does not hold, and let $\mu$ and $\mu^{\prime}$ be Stackelberg equilibrium policies with corresponding satisfaction probabilities $v$ and $v^{\prime}$. We have that $v = Tv \geq T_{\mu^{\prime}}v$. Composing $k$ times and taking the limit as $k$ tends to infinity, we have $v = \lim_{k \rightarrow \infty}{T^{k}v} \geq \lim_{k \rightarrow \infty}{T_{\mu^{\prime}}^{k}v} = v^{\prime}.$ By the same argument, $v^{\prime} \geq v$, implying that $v = v^{\prime}$ and thus uniqueness holds.
\end{proof}
By Lemma \ref{lemma:Stackelberg-Shapley}, we have that the satisfaction probability for some state $s$ can be computed as the linear combination of the satisfaction probabilities of its neighbor states, where the coefficients are the transition probabilities jointly determined by the control and adversary policies. Lemma \ref{lemma:Stackelberg-Shapley} provides us the potential to apply iterative algorithm to compute the satisfaction probability.
%%%%%%%%%%%%%%%%%%%%%%%%%%%%%%%%%%%%%%%%%%%%%%%%%%%%%%%%%%%%%%%
\subsection{Computing the Optimal Policy}
%%%%%%%%%%%%%%%%%%%%%%%%%%%%%%%%%%%%%%%%%%%%%%%%%%%%%%%%%%%%%%%
Motivated by model checking algorithms \cite{baier2008principles}, we first construct a product SG. Then we analyze Problem \ref{problem: max satisfaction probability} on the product SG. A product SG is defined as follows.

\begin{definition}\label{def: product SG}
\emph{(Product SG):} Given an SG $\mathcal{SG} = (S,U_C,U_A,Pr,\mathcal{L},\Pi)$ and a DRA $\mathcal{R}=(Q,\Sigma,\delta,q_0,\text{Acc})$, a (labeled) product SG is a tuple $\mathcal{G}=(S_\mathcal{G},U_C,U_A,Pr_\mathcal{G},\text{Acc}_\mathcal{G})$, where $S_\mathcal{G}=S\times Q$ is a finite set of states, $U_C$ is a finite set of control inputs, $U_A$ is a finite set of attack signals, $Pr_\mathcal{G}((s,q),u_C,u_A,(s^\prime,q^\prime))=Pr(s,u_C,u_A,s^\prime)$ if $\delta(q,\mathcal{L}(s^\prime))=q^\prime$, $\text{Acc}_\mathcal{G}=\{(L_\mathcal{G}(1),K_\mathcal{G}(1)),\allowbreak(L_\mathcal{G}(2), K_\mathcal{G}(2)),\cdots,(L_\mathcal{G}(Z),K_\mathcal{G}(Z))\}$ is a finite set of Rabin pairs such that $L_\mathcal{G}(z),K_\mathcal{G}(z)\subseteq S_\mathcal{G}$ for all $z=1,2,\cdots,Z$ with $Z$ being a positive integer. In particular, a state $(s,q)\in L_\mathcal{G}(z)$ if and only if $q\in L(z)$, and a state $(s,q)\in K_\mathcal{G}(z)$ if and only if $q\in K(z)$.
\end{definition}

By Definition \ref{def: CMDP} and Definition \ref{def: product SG}, we have the following observations. First, since the transition probability is determined by $\mathcal{SG}$ and the satisfaction condition is determined by $\mathcal{R}$, the satisfaction probability of $\phi$ on $\mathcal{SG}$ is equal to the satisfaction probability of $\phi$ on the product SG $\mathcal{G}$. Second, we can generate the corresponding path $s_0s_1\cdots$ on $\mathcal{SG}$ given a path $(s_0,q_0)(s_1,q_1)\cdots$ on the product SG $\mathcal{G}$. Finally, given a control policy $\mu$ synthesized on the product SG $\mathcal{G}$, a corresponding control policy $\mu_\mathcal{SG}$ on $\mathcal{SG}$ is obtained by letting $\mu_\mathcal{SG}(s_i)=\mu((s_i,q))$ for all time step $i$ \cite{wolff2012robust,baier2008principles}. Due to these one-to-one correspondence relationships, in the following, we analyze Problem \ref{problem: max satisfaction probability} on the product SG $\mathcal{G}$ and present an algorithm to compute the optimal control policy. When the context is clear, we use $s$ to represent state $(s,q)\in S_\mathcal{G}$.

We next introduce the concept of Generalized Accepting Maximal End Component (GAMEC), which is generalized from accepting maximal end component (AMEC) on MDP. 

\begin{definition}\label{def: subSG}
\emph{(Sub-SG):} A sub-SG of an SG $\mathcal{SG}= (S,U_C,U_A,Pr,s_0,\Pi,\mathcal{L})$ is a pair of states and actions $(C,D)$ where $\emptyset\neq C\subseteq S$ is a set of states, and $D:C\rightarrow 2^{U_C(s)}$ is an enabling function such that $D(s)\subseteq U_C(s)$ for all $s\in C$ and $\{s^\prime|Pr(s,u_C,u_A,s^\prime)>0,\forall u_A\in U_A(s),s\in C\}\subseteq C$.
\end{definition}

By Definition \ref{def: subSG}, we have a sub-SG is also an SG. Given Definition \ref{def: subSG}, a Generalized Maximal End Component (GMEC) is defined as follows.

\begin{definition}\label{def: GMEC}
A Generalized End Component (GEC) is a sub-SG $(C,D)$ such that the underlying digraph $G_{(C,D)}$ of sub-SG $(C,D)$ is strongly connected. A GMEC is a GEC $(C,D)$ such that there exists no other GEC $(C^\prime,D^\prime)\neq (C,D)$, where $C\subseteq C^\prime$ and $D(s)\subseteq D^\prime(s)$ for all $s\in C$.
\end{definition}
\begin{definition}\label{def: GAMEC}
A GAMEC on the product SG $\mathcal{G}$ is a GMEC if there exists some $(L_\mathcal{G}(z),K_\mathcal{G}(z))\in\text{Acc}_\mathcal{G}$ such that $L_\mathcal{G}(z)\cap C=\emptyset$ and $K_\mathcal{G}(z)\subseteq C$.
\end{definition}
%Denote the set of GAMECs as $\mathcal{C}$. Algorithm \ref{algo:GAMEC} is used to compute the set of GAMECs. Given a product SG $\mathcal{G}$, a set of GAMECs $\mathcal{C}$ can be initialized, with each $(C,D)\in\mathcal{C}$ initialized using $(L_\mathcal{G}(z),K_\mathcal{G}(z))\in\text{Acc}$ indicated in DRA $\mathcal{R}$. For each $(C,D)\in\mathcal{C}$, the states that should be added into GAMEC are those associated with the same strongly connected component (SCC) of the underlying digraph. Also, for each state $s\in C$, if $u_C$ satisfies $Pr(s,u_C,u_A,s^\prime)\geq 0$ for all $u_A$ with $s^\prime\in C$, then $u_C$ is enabled at state $s$. Given the set of GAMECs $\mathcal{C}=\{(C_1,D_1),\cdots,(C_h,D_h),\cdots,(C_{|\mathcal{C}|},D_{|\mathcal{C}|})\}$, the set of accepting states $\mathcal{E}$ is computed as $\mathcal{E}=\cup_{h=1}^{|\mathcal{C}|}C_h$. Using Algorithm \ref{algo:GAMEC}, the set of accepting states $\mathcal{E}$ can be computed.

By Definition \ref{def: GAMEC}, we have a set of states constitutes a GAMEC if there exists a control policy such that for any initial states in the GAMEC, the system remains in the GAMEC with probability one and the specification is satisfied with probability one. We denote the set of GAMECs as $\mathcal{C}$, and the set of states that constitute GAMEC as accepting states. Algorithm \ref{algo:GAMEC} is used to compute the set of GAMECs. Given a product SG $\mathcal{G}$, a set of GAMECs $\mathcal{C}$ can be initialized as $C=\emptyset$ and $D(s)=U_C(s)$ for all $s$. Also, we define a temporary set $\mathcal{C}_{temp}$ which is initialized as $\mathcal{C}_{temp}=S_\mathcal{G}$. Then from line 8 to line 17, we compute a set of states $R$ that should be removed from GMEC. The set $R$ is first initialized to be empty. Then for each state $s$ in each nontrivial strongly connected component (SCC) of the underlying diagraph, i.e., the SCC with more than one states, we modify the admissible actions at state $s$ by keeping the actions that can make the system remain in $C$ under any adversary action. If there exists no such admissible action at state $s$, then the state $s$ is added into $R$. From line 18 to line 26, we examine if there exists any state $s^\prime$ in current GMEC that will steer the system into states in $R$. In particular, by taking action $u_C$ at each state $s^\prime$, if there exists some adversary action $u_A$ such that the system is steered into some state $s\in R$, then $u_C$ is removed from $U_C(s^\prime)$. Moreover, if there exists no admissible action at state $s^\prime$, then $s^\prime$ is added to $R$. Then we update the GMEC set as shown from line 27 to line 32. This procedure is repeated until no further update can be made on GMEC set. Line 34 to line 40 is to find the GAMEC following Definition \ref{def: GAMEC}. Given the set of GAMECs $\mathcal{C}=\{(C_1,D_1),\cdots,(C_h,D_h),\cdots,(C_{|\mathcal{C}|},D_{|\mathcal{C}|})\}$ returned by Algorithm \ref{algo:GAMEC}, the set of accepting states $\mathcal{E}$ is computed as $\mathcal{E}=\cup_{h=1}^{|\mathcal{C}|}C_h$.

The main idea to computing the solution to \eqref{eq: LTL maximization} is to show that the max-min probability of \eqref{eq: LTL maximization} is equivalent to maximizing (over $\mu$) the worst-case probability of reaching the set of accepting states $\mathcal{E}$. Denote the probability of reaching the set of accepting states $\mathcal{E}$ as reachability probability. In the following, we formally prove the equivalence between the worst-case satisfaction probability of \eqref{eq: LTL maximization} and the worst-case reachability probability. Then, we present an efficient algorithm for computing a policy $\mu$ that maximizes the worst-case probability of reaching $\mathcal{E}$, with the proofs of the correctness and convergence of the proposed algorithm. In particular, our proposed solution is based on the following.

 \begin{center}
  	\begin{algorithm}
  	\small
  		\caption{Computing the set of GAMECs $\mathcal{C}$.}
  		\label{algo:GAMEC}
  		\begin{algorithmic}[1]
  			\Procedure{Compute\_GAMEC}{$\mathcal{G}$}
  			\State \textbf{Input}: Product SG $\mathcal{G}$
  			\State \textbf{Output:} Set of GAMECs $\mathcal{C}$
  			\State \textbf{Initialization:} Let $D(s)=U_C(s)$ for all $s\in S_\mathcal{G}$. Let $\mathcal{C}=\emptyset$ and $\mathcal{C}_{temp}=\{S_\mathcal{G}\}$
  			\Repeat 
  			\State $\mathcal{C}=\mathcal{C}_{temp}$, $\mathcal{C}_{temp}=\emptyset$
  			\For{$C\in\mathcal{C}$}
  			\State $R=\emptyset$ \Comment{$R$ is the set of states that should be removed}
  			\State Let $SCC_1,\cdots,SCC_n$ be the set of nontrivial strongly connected components (SCC) of the underlying diagraph $G_{(C,D)}$
  			\For{$i=1,\cdots,n$}
  			\For{each state $s\in SCC_i$}
  			\State $D(s)=\{u_C\in U_C(s)|s^\prime\in C~\text{where }Pr(s,u_C,u_A,s^\prime)>0,~\forall u_A\in U_A(s)\}$
  			\If{$D(s)=\emptyset$}
  			\State $R=R\cup\{s\}$
  			\EndIf
  			\EndFor
  			\EndFor
  			\While{$R\neq\emptyset$}
  			\State dequeue $s\in R$ from $R$ and $C$
  			\If{there exist $s^\prime\in C$ and $u_C\in U_C(s^\prime)$ such that $Pr(s^\prime,u_C,u_A,s)>0$ under some $u_A\in U_A(s^\prime)$}
  			\State $D(s^\prime)=D(s^\prime)\setminus\{u_C\}$
  			\If{$D(s^\prime)=\emptyset$}
  			\State $R=R\cup\{s^\prime\}$
  			\EndIf
  			\EndIf
  			\EndWhile
  			\For{$i=1,\cdots,n$}
  			\If{$C\cap SCC_i\neq\emptyset$}
  			\State $\mathcal{C}=\mathcal{C}_{temp}\cup\{C\cap SCC_i\}$
  			\EndIf
  			\EndFor
  			\EndFor
  			\Until{$\mathcal{C}=\mathcal{C}_{temp}$}
  			\For {$C\in\mathcal{C}$}
  			\For{$(L_\mathcal{G}(z),K_\mathcal{G}(z))\in\text{Acc}_\mathcal{G}$}
  			\If{$L_\mathcal{G}(z)\cap C\neq\emptyset$ or $K_\mathcal{G}(z)\not\subseteq C$}
  			\State $\mathcal{C}=\mathcal{C}\setminus C$
  			\EndIf
            \EndFor
  			\EndFor
  			\State \Return $\mathcal{C}$
  			\EndProcedure
  		  \end{algorithmic}
  	\end{algorithm}
  \end{center}

\begin{proposition}\label{PROPOSITION: REACHABILITY}
For any stationary control policy $\mu$ and initial state $s$, the minimum probability over all stationary adversary policies of satisfying the LTL formula is equal to the minimum probability over all stationary policies of reaching $\mathcal{E}$, i.e., given any stationary policy $\mu$, we have
\begin{equation}
\min_{\tau}Pr_{\mathcal{G}}^{\mu\tau}(\phi|s)=\min_{\tau}Pr_{\mathcal{G}}^{\mu\tau}(\mbox{reach }\mathcal{E}|s),
\end{equation}
where $Pr_{\mathcal{G}}^{\mu\tau}(\text{reach }\mathcal{E})$ is the probability of reaching $\mathcal{E}$ under policies $\mu$ and $\tau$.
%Given the stationary control policy $\mu$ from the controller and the stationary policy $\tau$ from the adversary, the probability of satisfying the LTL formula is equivalent to the probability of reaching the accepting maximal end components of $\mathcal{C}_p$.
\end{proposition}

\begin{proof}
By Definition of $\mathcal{E}$, if the system reaches $\mathcal{E}$, then $\phi$ is satisfied for a maximizing policy $\mu$. Thus $\min_{\tau}{Pr_{\mathcal{G}}^{\mu\tau}(\mbox{reach } \mathcal{E})} = \min_{\tau}{Pr_{\mathcal{G}}^{\mu\tau}(\phi)}.$

Suppose that for some control policy $\mu$ and initial state $s_{0}$,
\begin{equation}
\label{eq:prop1-1}
\min_{\tau}{Pr_{\mathcal{G}}^{\mu\tau}(\phi|s_{0})} > \min_{\tau}{Pr_{\mathcal{G}}^{\mu\tau}(\mbox{reach } \mathcal{E}|s_{0})},
\end{equation}
and let $\tau$ be a minimizing stationary policy for the adversary.  The policies $\mu$ and $\tau$ induce an MC on the state space. By model checking algorithms on MC \cite{baier2008principles}, the probability of satisfying $\phi$ from $s_{0}$ is equal to the probability of reaching a bottom strongly connected component (BSCC) that satisfies $\phi$. By assumption there exists a BSCC, denoted $SCC_{0}$, that is reachable from $s_{0}$, disjoint from $\mathcal{E}$, and yet satisfies $Pr_{\mathcal{G}}^{\mu\tau}(\phi | s) = 1$ for all $s \in S_{0}$ (if this were not the case, then (\ref{eq:prop1-1}) would not hold). %Hence the state $s_{0}$ satisfies $s_{0} \notin \mathcal{G}$ and yet $Pr(

Choose a state $s \in SCC_{0}$. Since $s \notin \mathcal{E}$, there exists a policy $\hat{\tau}$ such that $Pr_{\mathcal{P}}^{\mu\hat{\tau}}(\phi | s) < 1$. Create a new adversary policy $\tau_{1}$ as $\tau_{1}(s^{\prime}) = \hat{\tau}(s^{\prime})$ for all $ s^{\prime} \in SCC_{0}$ and
$\tau_{1}(s^{\prime}) =\tau(s^{\prime})$ otherwise. This policy induces a new MC on the state space. Furthermore, since only the outgoing transitions from $SCC_{0}$ are affected, the success probabilities of all sample paths that do not reach $SCC_{0}$ are unchanged.

If there exists any state $s^{\prime}$ that is reachable from $s$ in the new chain with $Pr_{\mathcal{G}}^{\mu\tau_{1}}(\phi|s^{\prime}) < 1$, then the policy $\tau_{1}$ strictly reduces the probability of satisfying $\phi$, thus contradicting the assumption that $\tau$ is a minimizing policy. Otherwise, let $SCC_{1}$ denote the set of states that are reachable from $s$ under $\mu$ and $\tau_{1}$ and are disjoint from $\mathcal{E}$ (this set must be non-empty; otherwise, the policy $\hat{\tau}$ would lead to $Pr_{\mathcal{G}}^{\mu\hat{\tau}}(\phi|s) = 1$, a contradiction). Construct a new policy $\tau_{2}$ by $\tau_{2}(s^{\prime}) = \hat{\tau}(s^{\prime})$ if $s^{\prime} \in S_{1}$ and $\tau_{2}(s^{\prime}) = \tau(s^{\prime})$ otherwise. Proceeding inductively, we derive a sequence of policies $\tau_{k}$ that satisfy $Pr_{\mathcal{G}}^{\mu\tau_{k}}(\phi) \leq Pr_{\mathcal{G}}^{\mu\tau}(\phi)$. This process terminates when either $Pr_{\mathcal{G}}^{\mu\tau_{k}}(\phi|s_{0}) < Pr_{\mathcal{G}}^{\mu\tau}(\phi|s_{0})$, contradicting the minimality of $\tau$, or when $Pr_{\mathcal{G}}^{\mu\tau_{k}}(\phi|s^{\prime\prime}) = Pr_{\mathcal{G}}^{\mu\hat{\tau}}(\phi|s^{\prime\prime})$ for all $s^{\prime\prime}$ that are reachable from $s$ under $\hat{\tau}$. The latter case, however, implies that $Pr_{\mathcal{G}}^{\mu\hat{\tau}}(\phi|s)) = 1$, contradicting the definition of $\hat{\tau}$. %there are no more states that are reachable from $s$ under $\hat{\tau}$ and yet $Pr_{\mathcal{C}}^{\mu\tau_{k}}(\phi|s) = 1$. This, however, contradicts the definition of $\hat{\tau}$,
%When the process terminates, then either $Pr_{\mathcal{C}}^{\mu\tau_{k}}(\phi | s) = 1$, thus contradicting the definition of $\hat{\tau}$, or $Pr_{\mathcal{C}}^{\mu\tau_{k}}(\phi | s) < 1$, contradicting minimality of $\tau$. These contradictions imply that the result of the proposition must hold.
%(\emph{Proof of Proposition \ref{PROPOSITION: REACHABILITY}})Note that given a stationary controller control policy $\mu$ and a stationary adversary control policy $\tau$, a Markov chain $\mathcal{MC}$ can be induced from CMDP. Using Theorem \ref{theorem:DRA analysis of MC}, we can see that Proposition \ref{PROPOSITION: REACHABILITY} holds.
%See Chapter $10.6$ in \cite{baier2008principles} for detailed proof.
\end{proof}

Proposition \ref{PROPOSITION: REACHABILITY} implies that the problem of maximizing the worst-case success probability can be mapped to a reachability problem on the product SG $\mathcal{G}$, where $\mathcal{G}$ is modified following Algorithm \ref{algo:redefineSG}. A dummy state $dest$ is added into the state space of $S_\mathcal{G}$. All transitions starting from a state in GAMECs are directed to state $dest$ with probability one regardless of the actions taken by the adversary. The transition probabilities and action spaces of all other nodes are unchanged. We observe that the reachability probability remains unchanged after applying Algorithm \ref{algo:redefineSG}. Hence the satisfaction probability remains unchanged. Moreover, the one-to-one correspondence of control policy still holds for states outside $\mathcal{E}$. Therefore, Problem \ref{problem: max satisfaction probability} is then equivalent to
\begin{equation}
\label{eq:equivalent_reachability}
\max_{\mu}{\min_{\tau}{Pr_{\mathcal{G}}^{\mu\tau}(\mbox{reach } {dest})}}
\end{equation}
Then, the solution to (\ref{eq: LTL maximization}) can be obtained from the solution to (\ref{eq:equivalent_reachability}) by following the optimal policy $\mu^{\ast}$ for (\ref{eq:equivalent_reachability}) at all states not in $\mathcal{E}$. The control policy for states in $\mathcal{E}$ can be any probability distribution over the set of enabled actions in each GAMEC.

 \begin{center}
  	\begin{algorithm}[!htp]
  		\caption{Modifying product SG $\mathcal{G}$.}
  		\label{algo:redefineSG}
  		\begin{algorithmic}[1]
  			\Procedure{Construct\_SG}{$\mathcal{G}$, $\mathcal{C}$}
  			\State \textbf{Input}: Product SG $\mathcal{G}$, the set of GAMECs $\mathcal{C}$
  			\State \textbf{Output:} Modified product SG $\mathcal{G}$
  			\State $S_\mathcal{G}:=S_\mathcal{G}\cup\{dest\}$,$U_C(s):=U_C(s)\cup\{d\},~\forall s\in S_\mathcal{G}$
  			\State $Pr_\mathcal{G}(s,d,u_A,dest)=1$ for all $s\in \mathcal{E}\cup\{dest\}$ and $u_A\in U_A(s)$
  			%\State $S_\mathcal{P}:=S_\mathcal{P}\cup\{\hat{s}_1,\cdots,\hat{s}_k\}\setminus\cup_{i=1}^kC_i$
  			\EndProcedure
  		\end{algorithmic}
  	\end{algorithm}
  \end{center}

Due to Proposition \ref{PROPOSITION: REACHABILITY}, in the following we focus on solving the problem \eqref{eq:equivalent_reachability}. Our approach for solving (\ref{eq:equivalent_reachability}) is to first compute a value vector $v \in \mathbb{R}^{|S_\mathcal{G}|}$, where $v(s) = \max_{\mu}{\min_{\tau}{Pr_{\mathcal{G}}^{\mu\tau}(\mbox{reach $dest$} | s)}}.$ By Lemma \ref{lemma:Stackelberg-Shapley}, the optimal policy can then be obtained from $v$ by choosing the distribution $\mu$ that solves the optimization problem of \eqref{eq:Stackelberg-Shapley} at each state $s$. Algorithm \ref{algo:reachability} gives a value iteration based algorithm for computing $v$. The idea of the algorithm is to initialize $v$ to be zero except on states in $\mathcal{E}$, and then greedily update $v(s)$ at each iteration by computing the optimal Stackelberg policy at each state. The algorithm terminates when a stationary $v$ is reached.

  \begin{center}
  	\begin{algorithm}[!htp]
  		\caption{Algorithm for a control strategy that maximizes the probability of satisfying $\phi$.}
  		\label{algo:reachability}
  		\begin{algorithmic}[1]
  			\Procedure{Max\_Reachability}{$\mathcal{G}$, $\mathcal{C}$}
  			\State \textbf{Input}: product SG $\mathcal{G}$, the set of GAMECs $\mathcal{C}$
  			\State \textbf{Output:} Vector $v \in \mathbb{R}^{|S_\mathcal{G}|}$, where $v(s) = \max{\min{Pr_{\mathcal{G}}^{\mu\tau}(\mbox{reach } dest | s_{0} = s)}}$
  			\State \textbf{Initialization:} $v^{0} \leftarrow 0$, $v^{1}(s) \leftarrow 1$ for $s \in \mathcal{E}$, $v^{1}(s) \leftarrow 0$ otherwise, $k \leftarrow 0$
  			\While{$\max{\{|v^{k+1}(s)-v^{k}(s)| : s \in S_\mathcal{G}\}} > \delta$}
  			\State $k \leftarrow k+1$
  			\For{$s \notin \mathcal{E}$}
  			\State Compute $v$ as $v^{k+1}(s) \leftarrow\allowbreak
  			\max_\mu\min_{\tau}\allowbreak\bigg\{\sum_{s^\prime}\sum_{u_{C} \in U_{C}(s)}\allowbreak\sum_{u_{A} \in U_{A}(s)}v(s^{\prime})\allowbreak
  			\mu(s,u_{C})\allowbreak\tau(s,u_{A})\allowbreak Pr_{\mathcal{G}}(s,u_{C},u_{A},s^{\prime})\bigg\}$
  			\EndFor
  			\EndWhile
  			\State \Return $v$
  			\EndProcedure        	
  		\end{algorithmic}
  	\end{algorithm}
  \end{center}

The following theorem shows that Algorithm \ref{algo:reachability} guarantees convergence to a Stackelberg equilibrium.

\begin{theorem}
\label{theorem:reachability_algo}
There exists $v^{\infty}$ such that for any $\epsilon > 0$, there exists $\delta$ and $K$ such that $||v^{k}-v^{\infty}||_{\infty} < \epsilon$ for $k > K$. Furthermore, $v^{\infty}$ satisfies the conditions of $v$ in Lemma \ref{lemma:Stackelberg-Shapley}.
\end{theorem}

\begin{proof}
We first show that, for each $s$, the sequence $v^{k}(s) : k=1,2,\ldots,$ is bounded and monotone. Boundedness follows from the fact that, at each iteration, $v^{k}(s)$ is a convex combination of the states of its neighbors, which are bounded above by $1$. To show monotonicity, we induct on $k$. Note that $v^{1}(s) \geq v^{0}(s)$ and $v^{2}(s) \geq v^{1}(s)$ since $v^{1}(s) = 0$ for $s \notin \mathcal{E}$ and $v^{k}(s) \equiv 1$ for $s \in \mathcal{E}$.

Let $\mu^{k}$ denote the optimal control policy at step $k$. We have
\begin{eqnarray}
\label{eq:reach-1}
v^{k+1}(s) &\geq& \min_\tau\sum_{u_{C} \in U_{C}(s)}\sum_{u_{A} \in U_{A}(s)}\sum_{s^{\prime} \in S}v^{k}(s^{\prime})\label{eq:reach-2}\\
&&\cdot\mu^{k}(s,u_{C})\tau(s,u_{A})Pr_\mathcal{G}(s,u_{C},u_{A},s^{\prime}) \nonumber\\
&\geq & \min_\tau\sum_{u_{C} \in U_{C}(s)}\sum_{u_{A} \in U_{A}(s)}\sum_{s^{\prime} \in S}v^{k-1}(s^{\prime})\\
\nonumber
&& \cdot \mu^{k}(s,u_{C})\tau(s,u_{A})Pr_\mathcal{G}(s,u_{C},u_{A},s^{\prime}) \\
\label{eq:reach-3}
&=& v^{k}(s)
\end{eqnarray}
Eq. (\ref{eq:reach-1}) follows because the value of $v^{k+1}(s)$, which corresponds to the maximizing policy, dominates the value achieved by the particular policy $\mu_{s}^{k}$. Eq. (\ref{eq:reach-2}) holds by induction, since $v^{k}(s^{\prime}) \geq v^{k-1}(s^{\prime})$ for all $s^{\prime}$. Finally, (\ref{eq:reach-3}) holds by construction of $\mu_{s}^{k}$. Hence $v^{k}(s)$ is monotone in $k$.

We therefore have that $v^{k}(s)$ is a bounded monotone sequence, and hence converges by the monotone convergence theorem. Let $v^{\infty}$ denote the vector of limit points, so that we can select $\delta$ sufficiently small (to prevent the algorithm from terminating before convergence) and $K$ large in order to satisfy $||v^{k}-v^{\infty}||_{\infty} < \epsilon$.

We now show that $v^{\infty}$ is a Stackelberg equilibrium. Since $v^{k}(s)$ converges, it is a Cauchy sequence and thus for any $\epsilon > 0$, there exists $K$ such that $k > K$ implies that $|v^{k}(s)-v^{k+1}(s)| < \epsilon$. By construction, this is equivalent to 
\begin{multline*}
\Big| v^{k}(s) - \max_{\mu}\min_{\tau}\sum_{u_{C} \in U_{C}(s)}\sum_{u_{A} \in U_{A}(s)}\sum_{s^{\prime} \in S}\big[v^{k-1}(s^{\prime})\mu(s,u_{C})\\
\tau(s,u_{A})Pr_\mathcal{G}(s,u_{C},u_{A},s^{\prime})\big]\Big| < \epsilon,
\end{multline*}
and hence $v^{\infty}$ is within $\epsilon$ of a Stackelberg equilibrium for every $\epsilon > 0$.
\end{proof}

While this approach guarantees asymptotic convergence to a Stackelberg equilibrium, there is no guarantee on the rate of convergence. By modifying Line 8 of the algorithm so that $v^{k+1}(s)$ is updated if
\begin{multline}\label{eq:epsilon-update}
\max_{\mu}\min_{\tau}\sum_{u_{C} \in U_{C}(s)}\sum_{u_{A} \in U_{A}(s)}\sum_{s^{\prime} \in S}\Big[v(s^{\prime})\mu(s,u_{C})
\tau(s,u_{A})\\Pr(s,u_{C},u_{A},s^{\prime})\Big] > (1+\epsilon)v^{k}(s)
\end{multline}
 and is constant otherwise, we derive the following result on the termination time.

\begin{proposition}
\label{prop:epsilon-complexity}
The $\epsilon$-relaxation of (\ref{eq:epsilon-update}) converges to a value of $v$ satisfying $\max\{|v^{k+1}(s)-v^{k}(s)| : s \in S_{\mathcal{G}}\} < \epsilon$  within $n\max_{s}{\left\{{\log{\left(\frac{1}{v^{0}(s)}\right)}}/{\log{(1+\epsilon)}}\right\}}$ iterations, where $v^{0}(s)$ is the smallest positive value of $v^{k}(s)$ for $k=0,1,\ldots$.
\end{proposition}

\begin{proof}
%We omit the proof due to space limit.
After $N$ updates, we have that $v^{N}(s) \geq (1+\epsilon)^{N}v^{0}(s)$. Hence for each $s$, $v(s)$ will be incremented at most $\max_{s}{\left\{{\log{\left(\frac{1}{v^{0}(s)}\right)}}/{\log{(1+\epsilon)}}\right\}}$ times. Furthermore, we have that at least one $v(s)$ must be updated at each iteration, thus giving the desired upper bound on the number of iterations. By definition of (\ref{eq:epsilon-update}), the set that is returned satisfies $|v^{k+1}(s) - v^{k}(s)| < \epsilon v^{k}(s) < \epsilon.$
\end{proof}

%Finally, we remark that the probability of any given adversary attack is relatively low. It may therefore be desirable to ensure that any policy $\mu$ also maximizes the probability of satisfying $\phi$ in the absence of any adversary input $U_{A}$. In general, such an optimal policy will consist of randomizing arbitrarily over a set of actions $U_{c}^{\ast}(s) \subseteq U_{c}(s)$ at each state $s$. Hence it suffices to follow the procedure discussed in this section using the reduced action spaces $\mathbf{U}_{c}^{\ast}$. We omit further discussion due to space limitations.

%%%%%%%%%%%%%%%%%%%%%%%%%%%%%%%%%%%%%%%%%%%%%%%%%%%%%%%%%%%%%%%
\section{Problem Formulation-Minimizing Invariant  Constraint Violation}\label{sec:formulation2}
%%%%%%%%%%%%%%%%%%%%%%%%%%%%%%%%%%%%%%%%%%%%%%%%%%%%%%%%%%%%%%%

In this section, we focus on a subclass of specifications of the form $\phi=\phi_1\land \psi$, where $\phi_1$ is an arbitrary LTL formula and $\psi$ is an invariant constraint. An invariant constraint requires the system to always satisfy some property. The general LTL formula $\phi_1$ can be used to model any arbitrary properties such as liveness $\phi_1=\Box\Diamond\pi$, while the invariant property can be used to model collision avoidance requirements $\psi=\Box\neg\text{obstacle}$. Given an LTL specification on the dynamical system \eqref{eq: dynamics}, it might be impossible for the system to satisfy the specification due to the presence of the adversary, i.e., $\max_{\mu}\min_{\tau}Pr_{\mathcal{SG}}^{\mu\tau}(\phi)=0$. Thus, we relax the specification by allowing violations on invariant constraint $\psi$. To minimize the impact of invariant constraint violations, we investigate the problem of minimizing the invariant constraint violation rate in this section. In particular, given a specification $\phi=\phi_1\land\psi$, the objective is to compute a control policy that minimizes the expected number of violations of $\psi$ per cycle over all the stationary policies that maximizes the probability of satisfying $\phi_1$. We say that every visit to a state that satisfies $\phi_1$ completes a cycle. We still focus on the SG generated from \eqref{eq: dynamics}.

In the following, we first formulate the problem. Motivated by the solution idea of ACPC problem, we solve the problem by generalizing the ACPS problem in \cite{bertsekas1995dynamic} and establishing a connection between the problem we formulated and the generalized ACPS problem. Finally, we present the optimality conditions of the problem of interest, and propose an efficient algorithm to solve the problem.

%%%%%%%%%%%%%%%%%%%%%%%%%%%%%%%%%%%%%%%%%%%%%%%%%%%%%%%%%%%%%%%
\subsection{Problem Statement}
%%%%%%%%%%%%%%%%%%%%%%%%%%%%%%%%%%%%%%%%%%%%%%%%%%%%%%%%%%%%%%%

In the following, we focus on how to generate a control policy that minimizes the rate at which $\psi$ is violated while guaranteeing that the probability of satisfying $\phi_1$ is maximized. The problem is stated as follows:

\begin{problem}\label{problem: P2}
Compute a secure control policy $\mu$ that minimizes the violation rate of $\psi$, i.e., the expected number of violations of $\psi$ per cycle, while maximizing the probability that $\phi_1$ is satisfied under any adversary policy $\tau$.
\end{problem}

To investigate the problem above, we assign a positive cost $\alpha$ to every transition initiated from a state $s$ if $s\not\models\psi$. If state $s\models\psi$, we let $g(s)=0$ for all $u_C$ and $u_A$. Thus we have for all $u_C$ and $u_A$
\begin{equation}\label{eq: transition cost}
g(s)=\begin{cases}
\alpha &\mbox{if }s\not\models \psi\\
0 &\mbox{if }s\models \psi.
\end{cases}
\end{equation}

By Proposition \ref{PROPOSITION: REACHABILITY}, we have that two consecutive visits to $\mathcal{E}$ complete a \emph{cycle}. Based on the transition cost defined in \eqref{eq: transition cost}, Problem \ref{problem: P2} can be rewritten as follows.
\begin{problem}\label{problem: ACPC}
Given a stochastic game $\mathcal{SG}$ and an LTL formula $\phi$ in the form of $\phi=\phi_1\land\psi$, obtain an optimal control policy $\mu$ that maximizes the probability of satisfying $\phi_1$ while minimizing the average cost per cycle due to violating $\psi$ which is defined as
\begin{equation}\label{eq: ACPC}
J_{\mathcal{SG}}^{\mu\tau}=\underset{N\rightarrow\infty}{\lim\sup}~\mathbb{E}\bigg\{\frac{\sum_{k=0}^Ng(s_k)}{I(\beta,N)}\mathrel{\Big|}
                \eta_\beta\models \phi_1\bigg\}.
\end{equation}
\end{problem}

Since $\phi_1$ is required to be satisfied, similar to our analysis in Section \ref{sec:formulation1}, we first construct a product SG $\mathcal{G}$ using SG $\mathcal{SG}$ and the DRA converted from specification $\phi_1$. Then we have the following observations. First, the one-to-one correspondence relationships between the control policies, paths, and associated expected cost due to violating $\psi$ on $\mathcal{SG}$ and $\mathcal{G}$ hold. Furthermore, we observe that if there exists a control policy such that the specification $\phi$ can be satisfied, it is the optimal solution to Problem \ref{problem: ACPC} with $J_{\mathcal{SG}}^{\mu\tau}=0$. Finally, by our analysis in Section \ref{sec:formulation1}, specification $\phi_1$ is guaranteed to be satisfied if there exists a control policy that can reach the set of accepting states $\mathcal{E}$. These observations provide us the advantage to analyze Problem \ref{problem: ACPC} on the product SG $\mathcal{G}$ constructed using SG $\mathcal{SG}$ and DRA constructed using $\phi_1$. Hence, in the following, we analyze Problem \ref{problem: ACPC} on the product SG $\mathcal{G}$. When the context is clear, we use $s$ to refer to state $(s,q)\in S_\mathcal{G}$. Without loss of generality, we assume that $\mathcal{E}=\{1,2,\cdots,l\}$, i.e., states $\{l+1,\cdots,n\}\cap \mathcal{E}=\emptyset$.
%In the following, we first generalize the ACPC problem discussed in \cite{ding2014optimal} and then solve Problem \ref{problem: ACPC} by establishing connections with the generalized ACPC problem.

%%%%%%%%%%%%%%%%%%%%%%%%%%%%%%%%%%%%%%%%%%%%%%%%%%%%%%%%%%%%%%%
\subsection{Computing the Optimal Control Policy}
%%%%%%%%%%%%%%%%%%%%%%%%%%%%%%%%%%%%%%%%%%%%%%%%%%%%%%%%%%%%%%%

%The main difference between Problem \ref{problem: ACPC} and the ACPC problem in \cite{ding2014optimal} is that we focus on the system with the presence of an adversary, while the formulation of ACPC problem is based on systems without adversaries. According to the analysis in \cite{ding2014optimal}, we have that the optimal control policy for ACPC problem can be mapped to the average cost per stage (ACPS) problem. ACPS problem focuses on minimizing the average cost incurred per stage, where each stage is defined as a visit to next state. ACPS problem has been extensively studied in \cite{bertsekas1995dynamic}.
Due to the presence of an adversary, the results presented in Proposition \ref{prop: ACPS} are not applicable. In the following we first generalize the ACPS problem discussed in \cite{bertsekas1995dynamic}, which focused on systems without adversaries. Then we characterize the optimality conditions for Problem \ref{problem: ACPC} by connecting it with the generalized ACPS problem.

%In the following, we first present the optimality conditions for ACPS problem in the presence of an adversary. Then based on the analysis in \cite{ding2014optimal}, we identify the connections between the ACPS problem and Problem \ref{problem: ACPC}. We finally derive the optimality conditions for the problem of minimizing safety constraint violation, which will be used to develop the algorithm for solving the problem.
\textbf{Generalized ACPS problem.} The presence of adversary is not considered in the ACPS problem considered in \cite{bertsekas1995dynamic}. Thus we need to formulate the ACPS problem with the presence of adversary and we denote it as the generalized ACPS problem. The objective of generalized ACPS problem is to minimize 
\begin{equation}
J_{\mu\tau}(s) = \underset{N\rightarrow\infty}{\limsup}~\frac{1}{N}\mathbb{E}\left\{\sum_{n = 0}^N g(s)\mid s_0 = s\right\}
\end{equation}
over all stationary control considering the adverary plays some strategy $\tau$ against the controller.

\textbf{Optimality conditions for generalized ACPS problem}. Given any stationary policies $\mu$ and $\tau$, denote the induced transition probability matrix as $P^{\mu\tau}$ with $P^{\mu\tau}(s,s^\prime)=\sum_{u_C\in U_C(s)}\sum_{u_A\in U_A(s)}\mu(s,u_C)\tau(s,u_A) Pr_\mathcal{G}(s,u_C,u_A,s^\prime)$. Analogously, denote the expected transition cost starting from any state $s\in S_\mathcal{G}$ as $g^{\mu\tau}(s)=\sum_{u_C\in U_C(s)}\sum_{u_A\in U_A(s)}\mu(s,u_C)\tau(s,u_A)g(s)$. Similar to \cite{bertsekas1995dynamic}, a gain-bias pair is used to characterize the optimality condition. The gain-bias pair $(B^{\mu\tau},b^{\mu\tau})$ under stationary policies $\mu$ and $\tau$, where $B^{\mu\tau}$ is the average cost per stage and $b^{\mu\tau}$ is the differential or relative cost vector, satisfies the following proposition.
\begin{lemma}\label{LEMMA: ACPS GAIN BIAS}
Let $\mu$ and $\tau$ be proper stationary policies for a communicating SG, where a communicating SG is an SG whose underlying graph is strongly connected. Then there exists a constant $\zeta^{\mu\tau}$ such that
\begin{equation}
B^{\mu\tau}(s)=\zeta^{\mu\tau},~\forall s\in S_\mathcal{G}.
\end{equation}
Furthermore, the gain-bias pair $(B^{\mu\tau},b^{\mu\tau})$ satisfies
\begin{equation}
B^{\mu\tau}(s)+b^{\mu\tau}(s)=g^{\mu\tau}(s)+\sum_{k=1}^n
P^{\mu\tau}(s,k)b^{\mu\tau}(k)
\end{equation}
\end{lemma}
\begin{proof}
Suppose $s^\prime$ is a recurrent state under policies $\mu$ and $\tau$. Define $\xi(s)$ as the expected cost to reach $s^\prime$ for the first time from state $s$, and $o(s)$ as the expected number of stages to reach $s^\prime$ for the first time from $s$. Thus $\xi(s^\prime)$ and $o(s^\prime)$ can be interpretated as the expected cost and expected number of stages to return to $s^\prime$ for the first time from state $s^\prime$, respectively. Based on the definitions above, we have the following equations:
\begin{align}
\xi(s)&=g^{\mu\tau}(s)+\sum_{k\in S\setminus s^\prime}P^{\mu\tau}(s,k)\xi(k),~\forall s\in S_\mathcal{G},\label{eq: ACPS expected cost}\\
o(s)&=1+\sum_{k\in S_\mathcal{G}\setminus s^\prime}P^{\mu\tau}(s,k)o(k),~\forall s\in S_\mathcal{G}.\label{eq: ACPS expected number of stages}
\end{align}
Define $\zeta^{\mu\tau}=\xi({s^\prime})/o({s^\prime})$. Multiplying \eqref{eq: ACPS expected number of stages} by $\zeta^{\mu\tau}$ and subtracting the associated product from \eqref{eq: ACPS expected cost}, we have
\begin{multline}
\xi(s)-\zeta^{\mu\tau}o(s)=g^{\mu\tau}(s)-\zeta^{\mu\tau}\\
+\sum_{k\in S_\mathcal{G}\setminus s^\prime}P^{\mu\tau}(s,k)(\xi(k)-\zeta^{\mu\tau}o(k)),~\forall s\in S_\mathcal{G}.\label{eq: ACPS rewrite}
\end{multline}
Define a bias term
\begin{equation}\label{eq: ACPS bias def}
b^{\mu\tau}(s)=\xi(s)-\zeta^{\mu\tau}o(s), ~\forall s\in S_\mathcal{G}
\end{equation}
Using \eqref{eq: ACPS bias def}, \eqref{eq: ACPS rewrite} can be rewritten as
\begin{equation}\label{eq: optimal ACPS gain bias}
\zeta^{\mu\tau}+b^{\mu\tau}(s)=g^{\mu\tau}(s)+\sum_{k=1}^nP^{\mu\tau}(s,k)b^{\mu\tau}(k),~\forall s\in S_\mathcal{G}
\end{equation}
which completes our proof.
\end{proof}

The result presented above generalizes the one in \cite{bertsekas1995dynamic} in the sense that we consider the presence of adversary. The reason that we focus on communicating SG is that we will focus on the accepting states which are strongly connected. Based on Lemma \ref{LEMMA: ACPS GAIN BIAS}, we have the optimality conditions for generalized ACPS problem expressed using the gain-bias pair $(B,b)$:
\begin{align}
&B(s)=\min_{\mu}\max_{\tau}\sum_{u_{C} \in U_{C}(s)}\sum_{u_{A} \in U_{A}(s)}\sum_{s^\prime}\mu(s,u_{C})\tau(s,u_{A})\nonumber\\
&\quad\quad\quad\quad\quad \cdot Pr_\mathcal{G}(s,u_C,u_A,s^\prime)B(s^\prime)\label{eq: ACPS gain}\\
&B(s)+b(s)=\min_{\mu\in\mu^*}\max_{\tau\in\tau^*}\big[g^{\mu\tau}(s)+\sum_{u_{C} \in U_{C}(s)}\sum_{u_{A} \in U_{A}(s)}\sum_{s^\prime}\nonumber\\
&\quad\quad\quad\mu(s,u_{C})\tau(s,u_{A}) Pr_\mathcal{G}(s,u_C,u_A,s^\prime)b(s^\prime)\bigg]\label{eq: ACPS gain bias}
\end{align}
where $\mu^*$ and $\tau^*$ are the optimal policy sets obtained by solving \eqref{eq: ACPS gain}. Eq. \eqref{eq: ACPS gain} can be shown using the method presented in Lemma \ref{lemma:Stackelberg-Shapley}, and \eqref{eq: ACPS gain bias} is obtained directly from \eqref{eq: optimal ACPS gain bias}. Given the optimality conditions \eqref{eq: ACPS gain} and \eqref{eq: ACPS gain bias} for generalized ACPS problem, we can derive the optimality conditions for Problem \ref{problem: ACPC} by mapping Problem \ref{problem: ACPC} to generalized ACPS problem.

\textbf{Optimality conditions for Problem \ref{problem: ACPC}}. In the following, we establish the connection between the generalized ACPS problem and Problem \ref{problem: ACPC}. Given the connection, we then derive the optimality conditions for Problem \ref{problem: ACPC}. 

Denote the gain-bias pair of Problem \ref{problem: ACPC} on the product SG $\mathcal{G}$ as $(J_{\mathcal{G}},h_{\mathcal{G}})$, where $J_{\mathcal{G}},h_{\mathcal{G}}\in\mathbb{R}^n$. Denote the gain-bias pair under policies $\mu$ and $\tau$ as $(J_{\mathcal{G}}^{\mu\tau},h_{\mathcal{G}}^{\mu\tau})$, where
$J_{\mathcal{G}}^{\mu\tau}=[J_{\mathcal{G}}^{\mu\tau}(1),J_{\mathcal{G}}^{\mu\tau}(2),\cdots,J_{\mathcal{G}}^{\mu\tau}(n)]^T$ and $ h_{\mathcal{G}}^{\mu\tau}=[h_{\mathcal{G}}^{\mu\tau}(1),h_{\mathcal{G}}^{\mu\tau}(2),\cdots,h_{\mathcal{G}}^{\mu\tau}(n)]^T.$

We can express the transition probability matrix $P^{\mu\tau}$ induced by control and adversary policy $\mu$ and $\tau$ as $P^{\mu\tau}=P^{\mu\tau}_{\text{in}}+P^{\mu\tau}_{\text{out}}$, where
\begin{subequations}\label{eq: probability matrix in out pi}
\begin{align}
P^{\mu\tau}_{\text{in}}(s,s^\prime)&=\begin{cases}
P^{\mu\tau}(s,s^\prime)&\mbox{if}~s^\prime\in \mathcal{E}\\
0&\mbox{otherwise}
\end{cases}\\
P^{\mu\tau}_{\text{out}}(s,s^\prime)&=\begin{cases}
P^{\mu\tau}(s,s^\prime)&\mbox{if}~s^\prime\notin \mathcal{E}\\
0&\mbox{otherwise}
\end{cases}.
\end{align}
\end{subequations}

Denote the probability that we visit some accepting state $s^\prime\in \mathcal{E}$ from state $s$ under policies $\mu$ and $\tau$ as $\hat{P}^{\mu\tau}(s,s^\prime)$. Then we see that  $\hat{P}^{\mu\tau}(s,s^\prime)$ is calculated as
\begin{multline}\label{eq: reach pi probability}
\hat{P}^{\mu\tau}(s,s^\prime)=\sum_{u_C\in U_C(s)}\mu(s,u_C)\sum_{u_A\in U_A(s)}\tau(s,u_A)\\
Pr_\mathcal{G}(s,u_C,u_A,s^\prime)
+\sum_{u_C\in U_C(s)}\mu(s,u_C)\sum_{u_A\in U_A(s)}\tau(s,u_A)\\
\sum_{k=l+1}^nPr_\mathcal{G}(s,u_C,u_A,k)\hat{P}^{\mu\tau}(k,s^\prime).
\end{multline}
The intuition behind \eqref{eq: reach pi probability} is that the probability that $s^\prime$ is the first state to be visited consists of the following two parts. The first term in \eqref{eq: reach pi probability} describes the probability that next state is in $\mathcal{E}$. The second term in \eqref{eq: reach pi probability} models the probability that before reaching state $s^\prime\in \mathcal{E}$, the next visiting state is $k\notin \mathcal{E}$. Denote the transition probability matrix formed by $\hat{P}(s,s^\prime)$ as $\hat{P}^{\mu\tau}$. Since $P^{\mu\tau}_{\text{out}}$ is substochastic and transient, we have $I-P^{\mu\tau}_{\text{out}}$ is non-singular \cite{hogben2013handbook}, where $I$ is the identity matrix with proper dimension. Thus $I-P^{\mu\tau}_{\text{out}}$ is invertible. Then using \eqref{eq: probability matrix in out pi}, the transition probability matrix $\hat{P}^{\mu\tau}$ is represented as
\begin{equation}\label{eq: rewrite reach pi probability}
\hat{P}^{\mu\tau}=(I-P^{\mu\tau}_{\text{out}})^{-1}P^{\mu\tau}_{\text{in}}.
\end{equation}

Denote the expected invariant property violation cost incurred when visiting some accepting state $s^\prime\in \mathcal{E}$ from state $s$ under policies $\mu$ and $\tau$ as $\hat{g}(s)$. The expected cost $\hat{g}(s)$ is calculated as follows:
\begin{equation}\label{eq: reach pi cost}
\hat{g}(s)=g^{\mu\tau}(s)+\sum_{k=l+1}^nPr^{\mu\tau}(s,k)\hat{g}(k).
\end{equation}
Under policies $\mu$ and $\tau$, denote the expected cost vector formed by $\hat{g}(s)$ as $\hat{g}^{\mu\tau}$. Then using \eqref{eq: probability matrix in out pi}, the expected cost vector \eqref{eq: reach pi cost} can be rearranged as follows:
\begin{equation}\label{eq: rewrite reach pi cost}
\hat{g}^{\mu\tau}= P^{\mu\tau}_{\text{out}}\hat{g}^{\mu\tau} +  g^{\mu\tau} \Rightarrow \hat{g}^{\mu\tau}=(I-P^{\mu\tau}_{\text{out}})^{-1}g^{\mu\tau}.
\end{equation}

Using \eqref{eq: rewrite reach pi probability} and \eqref{eq: rewrite reach pi cost}, we can rewrite \eqref{eq: ACPC} as
\begin{equation}\label{eq: rewrite ACPC}
J_{\mathcal{G}}^{\mu\tau}=\underset{N\rightarrow\infty}{\lim\sup}\frac{1}{N}\sum_{k=0}^{N-1}\hat{P}^{{\mu\tau}^k}\hat{g}^{\mu\tau}.
\end{equation}

Proper policies $\mu$ and $\tau$ of the product SG $\mathcal{G}$ for Problem \ref{problem: ACPC} are related to proper policies $\hat{\mu}$ and $\hat{\tau}$ for generalized ACPS problem as follows:
\begin{equation}\label{eq: connection}
\hat{P}^{\mu\tau}=P^{\hat{\mu}\hat{\tau}},\quad \hat{g}^{\mu\tau}=g^{\hat{\mu}\hat{\tau}},\quad J_{\mathcal{G}}^{\mu\tau}=B^{\hat{\mu}\hat{\tau}}.
\end{equation}
If we define a bias term $h_{\mathcal{G}}^{\mu\tau}=b^{\hat{\mu}\hat{\tau}}$, then a gain-bias pair $(J_{\mathcal{G}}^{\mu\tau},h_{\mathcal{G}}^{\mu\tau})$ is constructed for Problem \ref{problem: ACPC}. Under the worst case adversary policy $\hat{\tau}$, the control policy that makes the gain-bias pair of ACPS problem satisfy
\begin{equation}\label{eq: ACPS optimality condition}
B+b\leq g^{\hat{\mu}\hat{\tau}}+P^{\hat{\mu}\hat{\tau}}b
\end{equation}
is optimal. That is, the control policy $\mu^*$ that maps to $\hat{\mu}^*$ is optimal.

To obtain the optimal control policy, we need to characterize Problem \ref{problem: ACPC} in terms of the control and adversary policies $\mu$ and $\tau$. The following lemma generalizes the results presented in \cite{ding2014optimal} in which no adversary is considered. For completeness, we show its proof which generalizes the proof in \cite{ding2014optimal}.
\begin{lemma}\label{lemma: gain-bias calculation}
The gain-bias pair $(J_{\mathcal{G}}^{\mu\tau},h_{\mathcal{G}}^{\mu\tau})$ of Problem \ref{problem: ACPC} under policies $\mu$ and $\tau$ satisfies the following equations:
\begin{align}
J_{\mathcal{G}}^{\mu\tau}&=P^{\mu\tau}J_{\mathcal{G}}^{\mu\tau},\\
J_{\mathcal{G}}^{\mu\tau}+h_{\mathcal{G}}^{\mu\tau}&=g^{\mu\tau}+P^{\mu\tau}h_{\mathcal{G}}^{\mu\tau}+P^{\mu\tau}_{\text{out}}J_{\mathcal{G}}^{\mu\tau},\\
P^{\mu\tau}v^{\mu\tau}&=(I-P^{\mu\tau}_{\text{out}})h_{\mathcal{G}}^{\mu\tau}+v^{\mu\tau},
\end{align}
for some vector $v^{\mu\tau}$.
\end{lemma}
\begin{proof}
Given the policies $\hat{\mu}$ and $\hat{\tau}$ for ACPS problem, we have
\begin{align*}
J_{\mathcal{G}}^{\hat{\mu}\hat{\tau}}&=P^{\hat{\mu}\hat{\tau}}J_{\mathcal{G}}^{\hat{\mu}\hat{\tau}},\\
J_{\mathcal{G}}^{\hat{\mu}\hat{\tau}}+h_{\mathcal{G}}^{\hat{\mu}\hat{\tau}}&=g^{\hat{\mu}\hat{\tau}}+P^{\hat{\mu}\hat{\tau}}h_{\mathcal{G}}^{\hat{\mu}\hat{\tau}},\\
h_{\mathcal{G}}^{\hat{\mu}\hat{\tau}}+v^{\hat{\mu}\hat{\tau}}&=P^{\hat{\mu}\hat{\tau}}v^{\hat{\mu}\hat{\tau}},
\end{align*}
Due to the connection between the control policy of Problem \ref{problem: ACPC} and generalized ACPS problem, we have $$J_\mathcal{G}^{\mu\tau} = P^{\hat{\mu}\hat{\tau}}J_\mathcal{G}^{\mu\tau}=(I-P^{\mu\tau}_{\text{out}})^{-1}P^{\mu\tau}_{\text{in}}J_\mathcal{G}^{\mu\tau}.$$ By rearranging the equation above, we have $(I-P^{\mu\tau}_{\text{out}})J_\mathcal{G}^{\mu\tau}=J_\mathcal{G}^{\mu\tau}-P^{\mu\tau}_{\text{out}}J_\mathcal{G}^{\mu\tau}=P^{\mu\tau}_{\text{in}}J_\mathcal{G}^{\mu\tau}.$
Thus $J_\mathcal{G}^{\mu\tau}=(P^{\mu\tau}_{\text{out}}+P^{\mu\tau}_{\text{in}})J_\mathcal{G}^{\mu\tau}=P^{\mu\tau}J_\mathcal{G}^{\mu\tau}.$ The expression $J_{\mathcal{G}}^{\mu\tau}+h_{\mathcal{G}}^{\mu\tau}=g^{\mu\tau}+P^{\mu\tau}h_{\mathcal{G}}^{\mu\tau}+P^{\mu\tau}_{\text{out}}J_{\mathcal{G}}^{\mu\tau}$ can be rewritten using \eqref{eq: rewrite reach pi probability} and \eqref{eq: rewrite reach pi cost}. We have $$J_{\mathcal{G}}^{\mu\tau}+h_{\mathcal{G}}^{\mu\tau}=(I-P^{\mu\tau}_{\text{out}})^{-1}(g^{\mu\tau}+P^{\mu\tau}_{\text{in}}h_{\mathcal{G}}^{\mu\tau}).$$ Manipulating the equation above, we see that $(I-P^{\mu\tau}_{\text{out}})(J_{\mathcal{G}}^{\mu\tau}+h_{\mathcal{G}}^{\mu\tau})=g^{\mu\tau}+P^{\mu\tau}_{\text{in}}h_{\mathcal{G}}^{\mu\tau}.$ Then we can see that 
\begin{align*}
J_{\mathcal{G}}^{\mu\tau}+h_{\mathcal{G}}^{\mu\tau}&=g^{\mu\tau}+(P^{\mu\tau}_{\text{in}}+P^{\mu\tau}_{\text{out}})h_{\mathcal{G}}^{\mu\tau}+P^{\mu\tau}_{\text{out}}J_{\mathcal{G}}^{\mu\tau}\\
&=g^{\mu\tau}+P^{\mu\tau}h_{\mathcal{G}}^{\mu\tau}+P^{\mu\tau}_{\text{out}}J_{\mathcal{G}}^{\mu\tau}.
\end{align*}
Start from $h_{\mathcal{G}}^{\hat{\mu}\hat{\tau}}+v^{\hat{\mu}\hat{\tau}}=P^{\hat{\mu}\hat{\tau}}v^{\hat{\mu}\hat{\tau}}$. We see that $h_{\mathcal{G}}^{\hat{\mu}\hat{\tau}}+v^{\hat{\mu}\hat{\tau}}=(I-P^{\mu\tau}_{\text{out}})^{-1}P^{\mu\tau}_{\text{in}}v^{\mu\tau}.$ Therefore we have
\begin{equation*}
(I-P^{\mu\tau}_{\text{out}})h_{\mathcal{G}}^{\mu\tau}+v^{\mu\tau}=P^{\mu\tau}v^{\mu\tau},
\end{equation*}
which completes our proof.
\end{proof}
Lemma \ref{lemma: gain-bias calculation} indicates that the gain-bias pair can be solved as solutions to a linear system with $3n$ unknowns. Thus we can evaluate any control and adversary policies using Lemma \ref{lemma: gain-bias calculation}, which provides us the potential to implement iterative algorithm to compute the optimal control policy $\mu$.

%%%%%%%%%%%%%%%%%%%%%%%%%%%%%%%%%%%
\begin{figure*}[!t]
% ensure that we have normalsize text
\normalsize
% Store the current equation number. 
\setcounter{mytempeqncnt}{\value{equation}}
% Set the equation number to one less than the one
% desired for the first equation here.
% The value here will have to changed if equations
% are added or removed prior to the place these
% equations are referenced in the main text. 
\setcounter{equation}{35}
\begin{align}
&\left(T^*(J_{\mathcal{G}},h_{\mathcal{G}})\right)(s)=\min_{\mu}\max_{\tau}\bigg[\sum_{u_{C} \in U_{C}(s)}\sum_{u_{A} \in U_{A}(s)}\mu(s,u_C)\tau(s,u_{A})g(s)+\sum_{u_{C} \in U_{C}(s)}\sum_{u_{A} \in U_{A}(s)}\sum_{s^\prime=1}^n\mu(s,u_{C})\label{eq: optimal operator}\\
&\quad\quad\quad\quad\cdot\tau(s,u_{A})Pr_\mathcal{G}(s,u_C,u_A,s^\prime)h_{\mathcal{G}}(s^\prime)+\sum_{u_{C} \in U_{C}(s)}\sum_{u_{A} \in U_{A}(s)}\sum_{s^\prime=l+1}^n\mu(s,u_{C})\tau(s,u_{A})Pr_\mathcal{G}(s,u_C,u_A,s^\prime)J_{\mathcal{G}}(s^\prime)\bigg],~\forall s\nonumber\\
&\left(T_\mu(J_{\mathcal{G}},h_{\mathcal{G}})\right)(s)=\max_{\tau}\bigg[\sum_{u_{C} \in U_{C}(s)}\sum_{u_{A} \in U_{A}(s)}\mu(s,u_C)\tau(s,u_{A})g(s)+\sum_{u_{C} \in U_{C}(s)}\sum_{u_{A} \in U_{A}(s)}\sum_{s^\prime=1}^n\mu(s,u_{C})\tau(s,u_{A})\label{eq: regular operator}\\
&\quad\quad\quad\quad\cdot  Pr_\mathcal{G}(s,u_C,u_A,s^\prime)h_{\mathcal{G}}(s^\prime)+\sum_{u_{C} \in U_{C}(s)}\sum_{u_{A} \in U_{A}(s)}\sum_{s^\prime=l+1}^n\mu(s,u_{C})\tau(s,u_{A})Pr_\mathcal{G}(s,u_C,u_A,s^\prime)J_{\mathcal{SG}}(s^\prime)\bigg].~\forall s\nonumber
\end{align}
% Restore the current equation number. 
\setcounter{equation}{\value{mytempeqncnt}}
% IEEE uses as a separator
\hrulefill
% The spacer can be tweaked to stop underfull vboxes. 
\vspace*{4pt}
\end{figure*}

%%%%%%%%%%%%%%%%%%%%%%%%%%%%%%%%%%%
\addtocounter{equation}{2}
To compute the control policy $\mu$, we define two operators on $(J_{\mathcal{G}},h_{\mathcal{G}})$ in \eqref{eq: optimal operator} and \eqref{eq: regular operator}, denoted as $T^*(J_{\mathcal{G}},h_{\mathcal{G}})$ and $T(J_{\mathcal{G}},h_{\mathcal{G}})$. Generally speaking, we can view them as mappings from $(J_{\mathcal{G}},h_{\mathcal{G}})$ to $T^*(J_{\mathcal{G}},h_{\mathcal{G}})\in\mathbb{R}^n$ and $T_\mu(J_{\mathcal{G}},h_{\mathcal{G}})\in\mathbb{R}^n$, respectively. %In particular, the mappings we defined are respectively as follows:
% \begin{align}
% &\left(T^*(J_{\mathcal{G}},h_{\mathcal{G}})\right)(s)=\min_{\mu}\max_{\tau}\bigg[\sum_{u_{C} \in U_{C}(s)}\sum_{u_{A} \in U_{A}(s)}\\\label{eq: optimal operator}
% &\mu(s,u_C)\tau(s,u_{A})g(s,u_C,u_A)+\sum_{u_{C} \in U_{C}(s)}\sum_{u_{A} \in U_{A}(s)}\sum_{s^\prime=1}^n\nonumber\\
% &\mu(s,u_{C})\tau(s,u_{A})Pr_\mathcal{G}(s,u_C,u_A,s^\prime)
% h_{\mathcal{G}}(s^\prime)+\sum_{u_{C} \in U_{C}(s)}\sum_{u_{A} \in U_{A}(s)}\sum_{s^\prime=l+1}^n\nonumber\\
% &\mu(s,u_{C})\tau(s,u_{A})Pr_\mathcal{G}(s,u_C,u_A,s^\prime)J_{\mathcal{G}}(s^\prime)\bigg],~\forall s\nonumber\\
% &\left(T_\mu(J_{\mathcal{G}},h_{\mathcal{G}})\right)(s)=\max_{\tau}\bigg[\sum_{u_{C} \in U_{C}(s)}\mu(s,u_C)\sum_{u_{A} \in U_{A}(s)}\tau(s,u_{A})g(s,u_C,u_A)\label{eq: regular operator}\\
% &+\sum_{u_{C} \in U_{C}(s)}\sum_{u_{A} \in U_{A}(s)}\sum_{s^\prime=1}^n\mu(s,u_{C})\tau(s,u_{A})
% Pr_\mathcal{G}(s,u_C,u_A,s^\prime)h_{\mathcal{G}}(s^\prime)+\nonumber\\
% &\sum_{u_{C} \in U_{C}(s)}\sum_{u_{A} \in U_{A}(s)}\sum_{s^\prime=l+1}^n
% \mu(s,u_{C})\tau(s,u_{A})Pr_\mathcal{G}(s,u_C,u_A,s^\prime)J_{\mathcal{SG}}(s^\prime)\bigg].~\forall s\nonumber
% \end{align}
Note that in \eqref{eq: regular operator}, the transition probability is the one induced under a certain control policy $\mu$.

Based on the definitions above, we present the optimality conditions for Problem \ref{problem: ACPC} using the following theorem.
\begin{theorem}\label{theorem: optimality condition}
The control policy $\mu$ with gain-bias pair $(J_{\mathcal{G}}^{\mu\tau},h_{\mathcal{G}}^{\mu\tau})$ that satisfies
\begin{equation}\label{eq: optimality condition}
J_{\mathcal{G}}^{\mu\tau}+h_{\mathcal{G}}^{\mu\tau}=T^*(J_{\mathcal{G}}^{\mu\tau},h_{\mathcal{G}}^{\mu\tau})
\end{equation}
is the optimal control policy.
\end{theorem}
\begin{proof}
Consider any arbitrary control policy $\hat{\mu}$ and the worst case adversary policy $\hat{\tau}$. By definition of $T^*(\cdot)$ in \eqref{eq: optimal operator}, we have that \eqref{eq: optimality condition} implies $J_{\mathcal{G}}^{\mu\tau}+h_{\mathcal{G}}^{\mu\tau}\leq g^{\hat{\mu}\hat{\tau}}+P^{\hat{\mu}\hat{\tau}}h_{\mathcal{G}}^{\mu\tau}+P^{\hat{\mu}\hat{\tau}}_{\text{out}}J_{\mathcal{G}}^{\mu\tau},$ where $P^{\hat{\mu}\hat{\tau}}$ and $P^{\hat{\mu}\hat{\tau}}_{\text{out}}$ are the transition probability matrix induced by policies $\hat{\mu}$ and $\hat{\tau}$. Then we have 
\begin{align*}
J_{\mathcal{G}}^{\mu\tau}+h_{\mathcal{G}}^{\mu\tau}-P^{\hat{\mu}\hat{\tau}}_{\text{out}}J_{\mathcal{G}}^{\mu\tau}&\leq g^{\hat{\mu}\hat{\tau}}+P^{\hat{\mu}\hat{\tau}}h_{\mathcal{G}}^{\mu\tau}\\
&=g^{\hat{\mu}\hat{\tau}}+(P^{\hat{\mu}\hat{\tau}}_{\text{in}}+P^{\hat{\mu}\hat{\tau}}_{\text{out}})h_{\mathcal{G}}^{\mu\tau}.
\end{align*}
Thus we observe that $(I-P^{\hat{\mu}\hat{\tau}}_{\text{out}})(J_{\mathcal{G}}^{\mu\tau}+h_{\mathcal{G}}^{\mu\tau})\leq g^{\hat{\mu}\hat{\tau}}+P^{\hat{\mu}\hat{\tau}}_{\text{in}}h_{\mathcal{G}}^{\mu\tau}.$ Note that $(I-P^{\hat{\mu}\hat{\tau}}_{\text{out}})$ is invertible. Thus the inequality above is rewritten as $J_{\mathcal{G}}^{\mu\tau}+h_{\mathcal{G}}^{\mu\tau}\leq (I-P^{\hat{\mu}\hat{\tau}}_{\text{out}})^{-1}(g^{\hat{\mu}\hat{\tau}}+P^{\hat{\mu}\hat{\tau}}_{\text{in}}h_{\mathcal{G}}^{\mu\tau}).$ Rewrite the inequality above according to \eqref{eq: rewrite reach pi probability} and \eqref{eq: rewrite reach pi cost}. Then we have $$J_{\mathcal{G}}^{\mu\tau}+h_{\mathcal{G}}^{\mu\tau}\leq g^{\tilde{\mu}\tilde{\tau}}+P^{\tilde{\mu}\tilde{\tau}}h_{\mathcal{G}}^{\mu\tau},$$ where $\tilde{\mu}$ and $\tilde{\tau}$ are the control and adversary policies in the associated ACPS problem. Thus, $\tilde{\mu}^*$ satisfies \eqref{eq: ACPS optimality condition} and $\mu$ is optimal over all the proper policies.
\end{proof}

\begin{center}
\begin{algorithm}[!htp]
	\caption{Algorithm for a control strategy that minimizes the expected number of invariant constraint violations.}
	\label{algo:ACPC}
	\begin{algorithmic}[1]
		\Procedure{Min\_Violation}{$\mathcal{G}$, $\mathcal{C}$}
		\State \textbf{Input}: product SG $\mathcal{G}$, the set GAMECs $\mathcal{C}$ associated with formula $\phi_1$ 
		\State \textbf{Output:} Control policy $\mu_{\text{cycle}}$
		\State \textbf{Initialization:} Initialize $\mu^0$ and $\tau^0$ be proper policies.
        \While{$T^*(J_{\mathcal{C}}^{\mu^{k}\tau^k},h_{\mathcal{C}}^{\mu^{k}\tau^k})\neq T^*(J_{\mathcal{C}}^{\mu^{k-1}\tau^{k-1}},h_{\mathcal{C}}^{\mu^{k-1}\tau^{k-1}})$}		
        \State Policy Evaluation: Given $\mu^k$ and $\tau^k$, calculate the gain-bias pair $(J_{\mathcal{C}}^{\mu^{k}\tau^k},h_{\mathcal{C}}^{\mu^{k}\tau^k})$ using Lemma \ref{lemma: gain-bias calculation}.
        \State Policy Improvement: Calculate the control policy $\mu$ using $\mu=\argmin_{\mu}\>\argmax_{\tau}\left\{ g^{\mu\tau}+P^{\mu\tau}h_{\mathcal{C}}^{\mu^k\tau^k}+P^{\mu\tau}_{\text{out}}J_{\mathcal{C}}^{\mu^k\tau^k}\right\}$.
        \State Set $\mu^{k+1}=\mu$.
        \State Set $k = k + 1$.
        \EndWhile
        \EndProcedure        	
	 \end{algorithmic}
\end{algorithm}
\end{center}

\textbf{Optimal control policy for Problem \ref{problem: ACPC}}. In the following, we focus on how to obtain an optimal secure control policy. First, note that the optimal control policy consists of two parts. The first part, denoted as $\mu_{\text{reach}}$, maximizes the probability of satisfying specification $\phi_1$, while the second part, denoted as $\mu_{\text{cycle}}$, minimizes the violation cost per cycle due to violating invariant property. Following the procedure described in Algorithm \ref{algo:reachability}, we can obtain the control policy $\mu_{\text{reach}}$ that maximizes the probability of satisfying specification $\phi_1$. Suppose the set of accepting states $\mathcal{E}$ has been reached. Then the control policy $\mu_{\text{cycle}}$ that optimizes the long term performance of the system is generated using Algorithm \ref{algo:ACPC}. Algorithm \ref{algo:ACPC} first initializes the control and adversary policies arbitrarily (e.g., if $\mu^0$ and $\tau^0$ are set as uniform distributions, then $\mu^0(s,u_C)=1/|U_C(s)|$ and $\tau^0(s,u_A)=1/|U_A(s)|$ for all $s$, $u_C$ and $u_A$). Then it follows a policy iteration procedure to update the control and the corresponding adversary policies until no more improvement can be made. Given $\mu_{\text{reach}}$ and $\mu_{\text{cycle}}$, we can construct the optimal control policy for Problem \ref{problem: ACPC} as
\begin{equation}
\mu^*=\begin{cases}
\mu_{\text{reach}},\quad \text{if}~s\notin \mathcal{E}\\
\mu_{\text{cycle}},\quad \text{if}~s\in \mathcal{E}
\end{cases}.
\end{equation}

We finally present the convergence and optimality of Algorithm \ref{algo:ACPC} using the following theorem. 
\begin{theorem}\label{theom: correctness of algo 5}
Algorithm \ref{algo:ACPC} terminates within a finite number of iterations for any given accepting state set $\mathcal{E}$. Moreover, the result returned by Algorithm \ref{algo:ACPC} satisfies the optimality conditions for Problem \ref{problem: ACPC}.
\end{theorem}
\begin{proof}
In the following, we first prove Algorithm \ref{algo:ACPC} converges within a finite number of iterations. Then we prove that the results returned by Algorithm \ref{algo:ACPC} satisfies the optimality conditions in Theorem \ref{theorem: optimality condition}. We denote the iteration index as $k$. The control policy at $k$th iteration is denoted as $\mu^k$. The worst case adversary policy associated with $\mu^k$ is denoted as $\tau^k$. Define a vector $\delta\in\mathbb{R}^n$ as $\delta = J_{\mathcal{G}}^{\mu^k\tau^k}\mathbf{1}+h_{\mathcal{G}}^{\mu^k\tau^k}
-g^{\mu^{k+1}\tau^{k+1}}-P^{\mu^{k+1}\tau^{k+1}}h_{\mathcal{G}}^{\mu^k\tau^k}-P^{\mu^{k+1}\tau^{k+1}}_{\text{out}}J_{\mathcal{G}}^{\mu^k\tau^k}\mathbf{1}.$

By Lemma \ref{lemma: gain-bias calculation}, we have $J_{\mathcal{G}}^{\mu^k\tau^k}\mathbf{1}+h_{\mathcal{G}}^{\mu^k\tau^k} = g^{\mu\tau}+P^{\mu\tau}h_{\mathcal{G}}^{\mu^k\tau^k}+P^{\mu\tau}_{\text{out}}J_{\mathcal{G}}^{\mu^k\tau^k}.$ By the definition of $T^*(\cdot)$ in \eqref{eq: optimal operator}, the control policy at iteration $k+1$ is computed by optimizing $g^{\mu\tau}+P^{\mu\tau}h_{\mathcal{G}}^{\mu^k\tau^k}+P^{\mu\tau}_{\text{out}}J_{\mathcal{G}}^{\mu^k\tau^k}$. Thus we have that for all $s$, $\delta(s)\geq 0$. Moreover, we can rewrite vector $\delta$ as
\begin{align*}
\delta&=J_{\mathcal{G}}^{\mu^k\tau^k}\mathbf{1}+h_{\mathcal{G}}^{\mu^k\tau^k}-g^{\mu^{k+1}\tau^{k+1}}-P^{\mu^{k+1}\tau^{k+1}}h_{\mathcal{G}}^{\mu^{k+1}\tau^{k+1}}\\
&\quad-P^{\mu^{k+1}\tau^{k+1}}_{\text{out}}J_{\mathcal{G}}^{\mu^{k+1}\tau^{k+1}}\mathbf{1}
+P^{\mu^{k+1}\tau^{k+1}}h_{\mathcal{G}}^{\mu^{k+1}\tau^{k+1}}\\
&\quad+P^{\mu^{k+1}\tau^{k+1}}_{\text{out}}J_{\mathcal{G}}^{\mu^{k+1}\tau^{k+1}}\mathbf{1}
-P^{\mu^{k+1}\tau^{k+1}}h_{\mathcal{G}}^{\mu^k\tau^k}\\
&\quad-P^{\mu^{k+1}\tau^{k+1}}_{\text{out}}J_{\mathcal{G}}^{\mu^k\tau^k}\mathbf{1}\\
&=J_{\mathcal{G}}^{\mu^k\tau^k}\mathbf{1}+h_{\mathcal{G}}^{\mu^k\tau^k}-J_{\mathcal{G}}^{\mu^{k+1}\tau^{k+1}}\mathbf{1}-h_{\mathcal{G}}^{\mu^{k+1}\tau^{k+1}}\\ &\quad-P^{\mu^{k+1}\tau^{k+1}}\left(h_{\mathcal{G}}^{\mu^k\tau^k}-h_{\mathcal{G}}^{\mu^{k+1}\tau^{k+1}}\right)\\
&\quad-P^{\mu^{k+1}\tau^{k+1}}_{\text{out}}\left(J_{\mathcal{G}}^{\mu^k\tau^k}-J_{\mathcal{G}}^{\mu^{k+1}\tau^{k+1}}\right)\mathbf{1},
\end{align*}
where the second equality holds by Lemma \ref{lemma: gain-bias calculation}. Thus $\delta$ can be represented as
\begin{multline}\label{eq: ACPC difference}
\delta=\left(I-P^{\mu^{k+1}\tau^{k+1}}_{\text{out}}\right)\left(J_{\mathcal{G}}^{\mu^k\tau^k}-J_{\mathcal{G}}^{\mu^{k+1}\tau^{k+1}}\right)\mathbf{1}
\\+\left(I-P^{\mu^{k+1}\tau^{k+1}}\right)\left(h_{\mathcal{G}}^{\mu^k\tau^k}-h_{\mathcal{G}}^{\mu^{k+1}\tau^{k+1}}\right),
\end{multline}
where $I$ is the identity matrix. By multiplying ${P^{\mu^{k+1}\tau^{k+1}}}^t$ to both sides of \eqref{eq: ACPC difference} and calculating the summation over $t$ from $0$ to $T-1$, we have that
\begin{multline}\label{eq: infinity sum}
\sum_{t=0}^{T-1}{P^{\mu^{k+1}\tau^{k+1}}}^t\delta =  \sum_{t=0}^{T-1}{P^{\mu^{k+1}\tau^{k+1}}}^t\left(I-P^{\mu^{k+1}\tau^{k+1}}_{\text{out}}\right)\\
\cdot\left(J_{\mathcal{G}}^{\mu^k\tau^k}-J_{\mathcal{G}}^{\mu^{k+1}\tau^{k+1}}\right)\mathbf{1}
+\sum_{t=0}^{T-1}{P^{\mu^{k+1}\tau^{k+1}}}^t\left(I-P^{\mu^{k+1}\tau^{k+1}}\right)\\
\cdot\left(h_{\mathcal{G}}^{\mu^k\tau^k}-h_{\mathcal{G}}^{\mu^{k+1}\tau^{k+1}}\right).
\end{multline}
Divide both sides by $T$ and let $T\rightarrow\infty$. Then we have
\begin{multline}\label{eq: gain difference}
\lim_{T\rightarrow\infty}\sum_{t=0}^{T-1}\frac{1}{T}{P^{\mu^{k+1}\tau^{k+1}}}^t\delta=
\lim_{T\rightarrow\infty}\sum_{t=0}^{T-1}\frac{1}{T}
\big({P^{\mu^{k+1}\tau^{k+1}}}^t\\-{P^{\mu^{k+1}\tau^{k+1}}}^tP^{\mu^{k+1}\tau^{k+1}}_{\text{out}}\big)
\big(J_{\mathcal{G}}^{\mu^k\tau^k}-J_{\mathcal{G}}^{\mu^{k+1}\tau^{k+1}}\big)\mathbf{1}
\end{multline}
since the second term of (\ref{eq: infinity sum}) is eliminated when $T\rightarrow\infty$. 
Since $P^{\mu^{k+1}\tau^{k+1}}_{\text{out}}$ is a substochastic matrix, we have $P^{\mu^{k+1}\tau^{k+1}}_{\text{out}}\mathbf{1}\leq \mathbf{1}$. Furthermore, since $P^{\mu^{k+1}\tau^{k+1}}$ is a stochastic matrix, we see that $\mathbf{1}-P^{\mu^{k+1}\tau^{k+1}}_{\text{out}}\mathbf{1}\geq 0$. Thus we have $\left({P^{\mu^{k+1}\tau^{k+1}}}^t-{P^{\mu^{k+1}\tau^{k+1}}}^tP^{\mu^{k+1}\tau^{k+1}}_{\text{out}}\right)\mathbf{1}\geq 0.$ Given the inequality above and $\delta \geq 0$, we have that $J_{\mathcal{G}}^{\mu^k\tau^k}-J_{\mathcal{G}}^{\mu^{k+1}\tau^{k+1}}\geq 0$ by observing \eqref{eq: gain difference}, which implies that $J_{\mathcal{G}}^{\mu^k\tau^k}\geq J_{\mathcal{G}}^{\mu^{k+1}\tau^{k+1}}$.

Consider the scenario where $J_{\mathcal{G}}^{\mu^k\tau^k}= J_{\mathcal{G}}^{\mu^{k+1}\tau^{k+1}}$. We further need to show that in this case $h_{\mathcal{G}}^{\mu^k\tau^k}\leq h_{\mathcal{G}}^{\mu^{k+1}\tau^{k+1}}$. For each state that belongs to the recurrent class, the corresponding entry of $\sum_{t=0}^{T-1}{P^{\mu^{k+1}\tau^{k+1}}}^t$ is positive. By observing \eqref{eq: gain difference}, we have $\delta(s)=0$ for all $s$ belonging to the recurrent class. Thus according to \eqref{eq: infinity sum}, we have that $h_{\mathcal{G}}^{\mu^k\tau^k}(s)= h_{\mathcal{G}}^{\mu^{k+1}\tau^{k+1}}(s)$ for all $s$ in the recurrent class. 

By observing \eqref{eq: infinity sum}, we have \begin{align*}
&\lim_{T\rightarrow\infty}\sum_{t=0}^{T-1}{P^{\mu^{k+1}\tau^{k+1}}}^t(h_{\mathcal{G}}^{\mu^k\tau^k}- h_{\mathcal{G}}^{\mu^{k+1}\tau^{k+1}})\\
=&\quad h_{\mathcal{G}}^{\mu^k\tau^k}- h_{\mathcal{G}}^{\mu^{k+1}\tau^{k+1}}-\lim_{T\rightarrow\infty}\sum_{t=0}^{T-1}{P^{\mu^{k+1}\tau^{k+1}}}^t\delta\\
\leq&\quad h_{\mathcal{G}}^{\mu^k\tau^k}- h_{\mathcal{G}}^{\mu^{k+1}\tau^{k+1}}-\delta.
\end{align*}
Note that the elements corresponding to the transient states in ${P^{\mu^{k+1}\tau^{k+1}}}^t(h_{\mathcal{G}}^{\mu^k\tau^k}- h_{\mathcal{G}}^{\mu^{k+1}\tau^{k+1}})$ approach zero when $t\rightarrow\infty$. Thus we have $h_{\mathcal{G}}^{\mu^k\tau^k}(s)- h_{\mathcal{G}}^{\mu^{k+1}\tau^{k+1}}(s)\geq \delta(s)\geq 0$ for all transient states $s$. Combining all the above together, we have that $\mu^k=\mu^{k+1}$ when $\delta=0$, otherwise $h_{\mathcal{G}}^{\mu^k\tau^k}(s)- h_{\mathcal{G}}^{\mu^{k+1}\tau^{k+1}}(s)\geq 0$ holds for some transient state $s$. %Therefore, we have that $T^*(J_{\mathcal{C}}^{\mu\tau},h_{\mathcal{C}}^{\mu\tau})$ decreases with respect to iteration in Algorithm \ref{algo:ACPC} in the sense that either $J_{\mathcal{C}}^{\mu\tau}$ or $h_{\mathcal{C}}^{\mu\tau}$ decreases.

When Algorithm \ref{algo:ACPC} terminates, we have that
\begin{equation}\label{eq: terminate T calculation}
T^*(J_{\mathcal{G}}^{\mu^{k+1}\tau^{k+1}},h_{\mathcal{G}}^{\mu^{k+1}\tau^{k+1}})=T^*(J_{\mathcal{G}}^{\mu^k\tau^k},h_{\mathcal{G}}^{\mu^k\tau^k}).
\end{equation}
By using policy iteration algorithm, the gain-bias pair $(J_{\mathcal{G}}^{\mu^k\tau^k},h_{\mathcal{G}}^{\mu^k\tau^k})$ is first evaluated using Lemma \ref{lemma: gain-bias calculation} at each iteration $k$. Then using the gain-bias pair obtained in policy evaluation phase, the $T^*$ operator is calculated as shown in Algorithm \ref{algo:ACPC}. Thus according to Lemma \ref{lemma: gain-bias calculation}, we see
\begin{equation}\label{eq: T operator calculation}
\mu=\argmin_{\mu}\>\max_{\tau}\left\{ g^{\mu\tau}+P^{\mu\tau}h_{\mathcal{G}}^{\mu^k\tau^k}+P^{\mu\tau}_{\text{out}}J_{\mathcal{G}}^{\mu^k\tau^k}\right\}.
\end{equation}
Note that the right hand side of \eqref{eq: T operator calculation} is equivalent to how $T^*$ is calculated in Algorithm \ref{algo:ACPC}. Therefore, by combining \eqref{eq: terminate T calculation} and \eqref{eq: T operator calculation}, we obtain $J_{\mathcal{G}}^{\mu^{k}\tau^{k}}+h_{\mathcal{G}}^{\mu^{k}\tau^{k}}=T^*(J_{\mathcal{G}}^{\mu^k\tau^k},h_{\mathcal{G}}^{\mu^k\tau^k}).$ By Theorem \ref{theorem: optimality condition}, we see that $\mu^{k}$ is the optimal control policy.
\end{proof}

% \begin{remark}
% Algorithm \ref{algo:ACPC} and Theorem \ref{theom: correctness of algo 5} are general in the sense that they are applicable to non-adversarial cases such as \cite{ding2014optimal} but not vice versa.
% \end{remark}

%%%%%%%%%%%%%%%%%%%%%%%%%%%%%%%%%%%%%%%%%%%%%%%%%%%%%%%%%%%%%%%
%\subsection{Extensive Discussion}
%%%%%%%%%%%%%%%%%%%%%%%%%%%%%%%%%%%%%%%%%%%%%%%%%%%%%%%%%%%%%%%

%Note that by modifying the definition of transition cost function \eqref{eq: transition cost}, variants of Problem \ref{problem: ACPC} can be developed. For instance, Problem \ref{problem: ACPC} can be modeled to minimize the trade-off between safety and liveness violations, or to maximize the time untile a safety constraint is violated. Furthermore, the control strategy using our proposed method is optimal in the absence adversary inputs. In this case, the problem becomes a stochastic game with single controller, which can be efficiently sovled \cite{filar2012competitive}.

\section{Case Study}\label{sec:simulation}

\begin{figure*}[t!]
\centering
                 \begin{subfigure}{.33\textwidth}
                 \includegraphics[width=\textwidth]{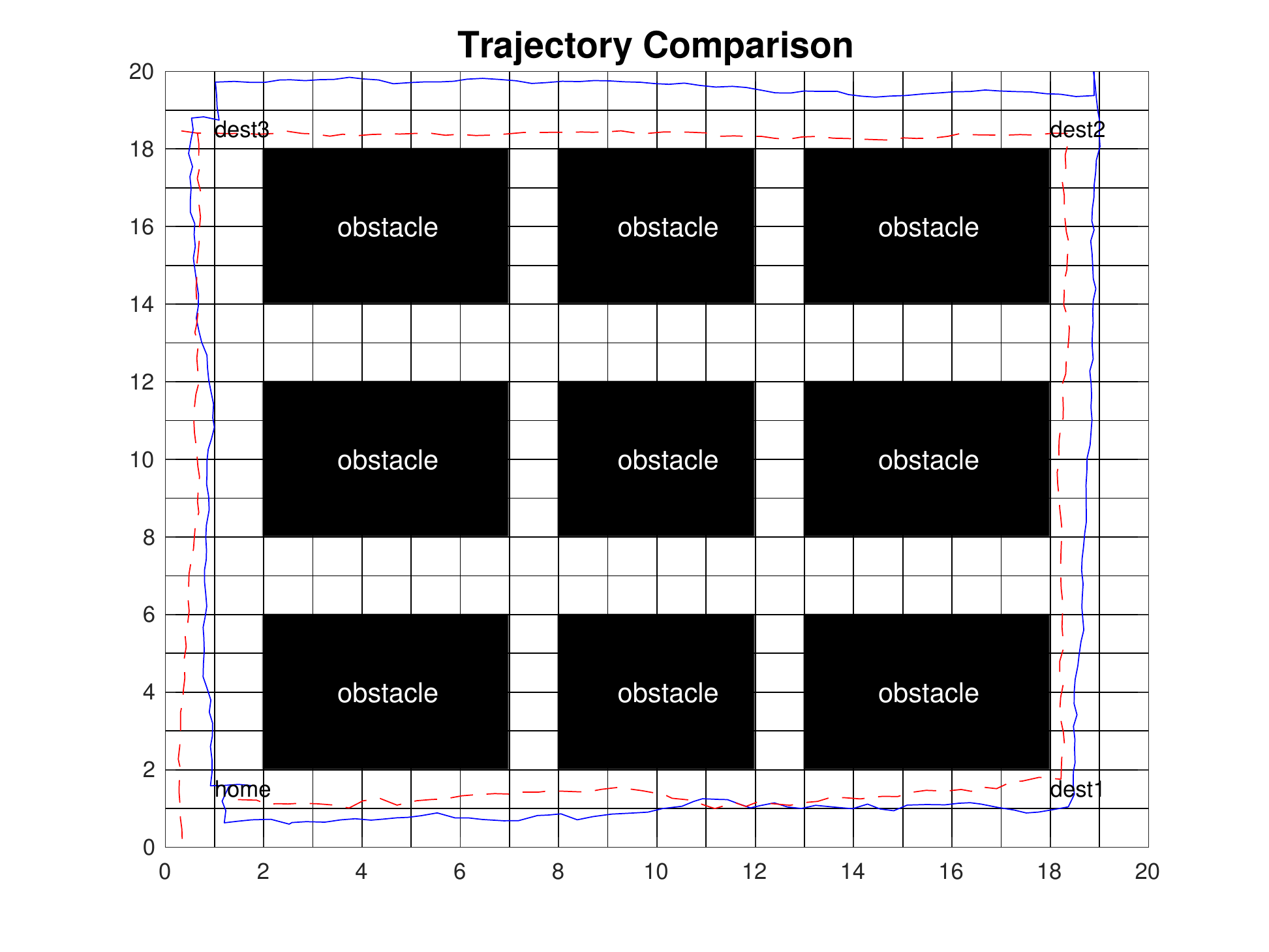}
                 \subcaption {}
                 \label{fig:traj}
                 \end{subfigure}\hfill
                 \begin{subfigure}{.33\textwidth}
                 \includegraphics[width=\textwidth]{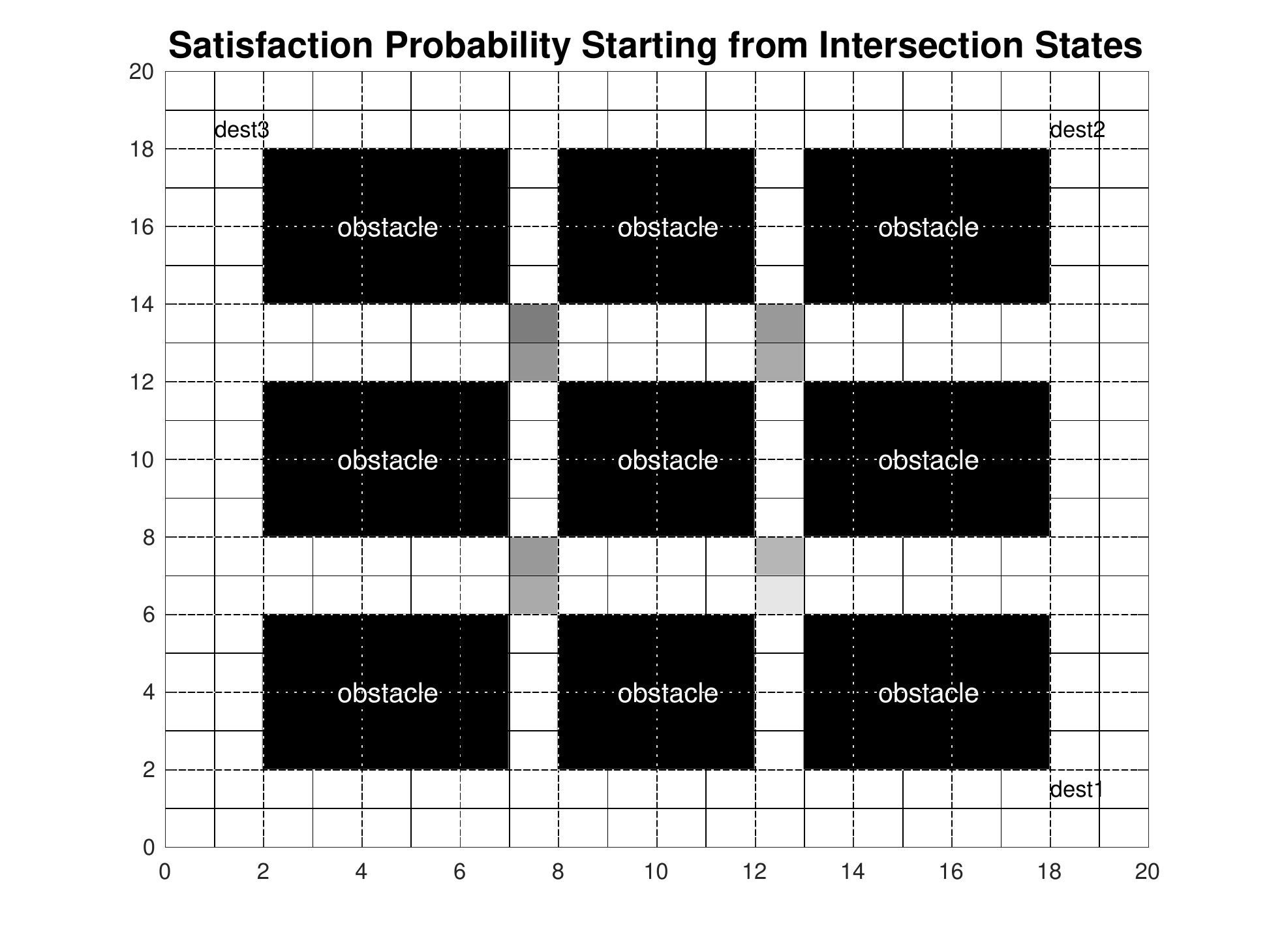}
                 \subcaption {}
                 \label{fig:traj1}
                 \end{subfigure}\hfill
                 \begin{subfigure}{.33\textwidth}
                 \includegraphics[width=\textwidth]{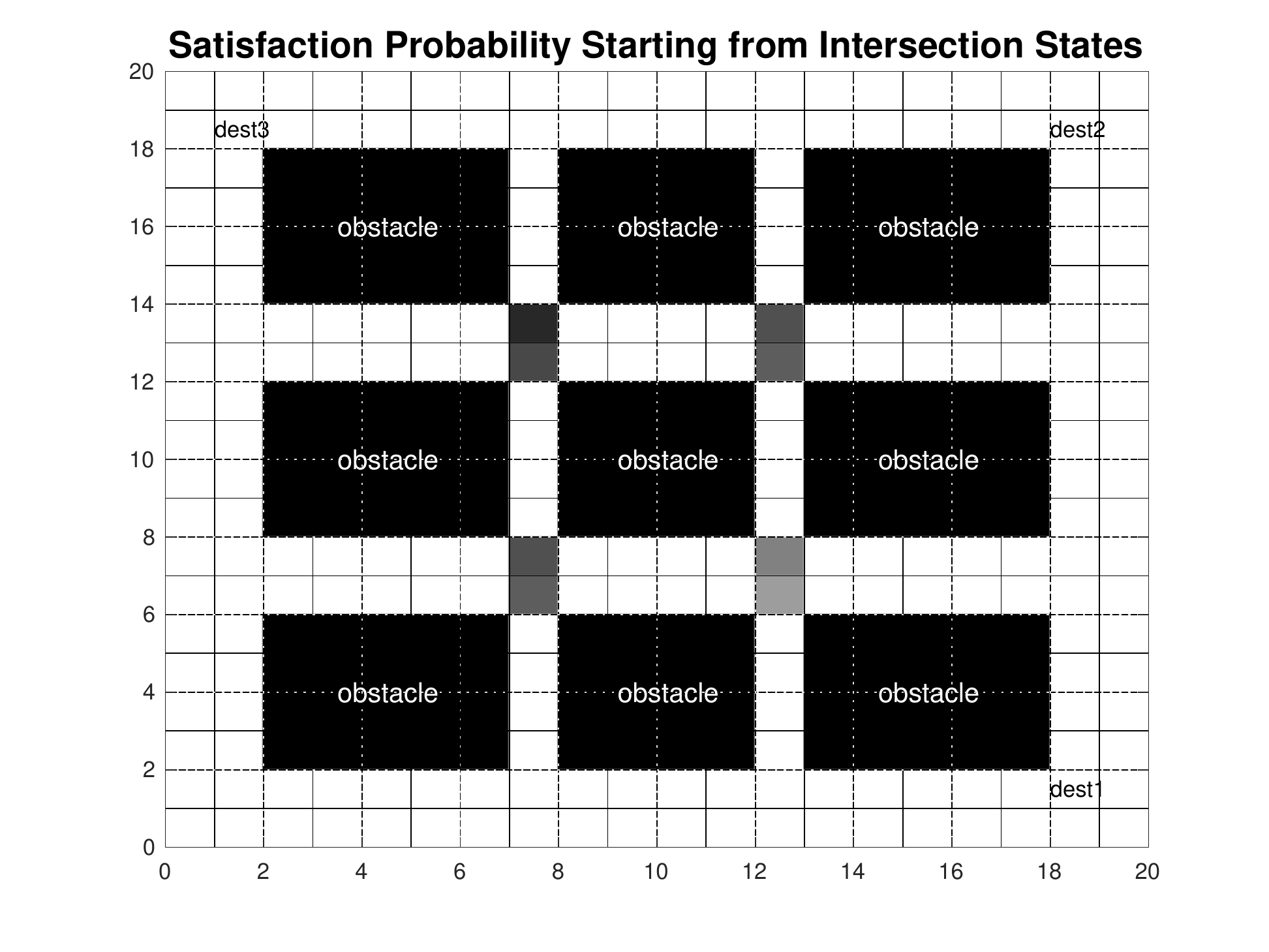}
                 \subcaption{}
                 \label{fig:traj2}
                 \end{subfigure}\hfill
\caption{Comparison of the proposed approach and the approach without considering the presence of the adversary. Fig. \ref{fig:traj} gives the trajectories obtained using two approaches. The solid blue line is the trajectory obtained using the proposed approach, while the dashed red line represents the trajectory obtained using the approach without considering the presence of the adversary. Fig. \ref{fig:traj1} and Fig. \ref{fig:traj2} present the probability of satisfying the LTL specification using the proposed approach and the approach without considering the adversary when the initial state is set as each of the states lands in the intersections of the grid world, respectively. The shade of gray level at the intersection states corresponds to the satisfaction probability, with black being probability $0$ and white being probability $1$.}
\end{figure*}

\begin{figure*}[t!]
\centering
                 \begin{subfigure}{.43\textwidth}
                 \includegraphics[width=\textwidth]{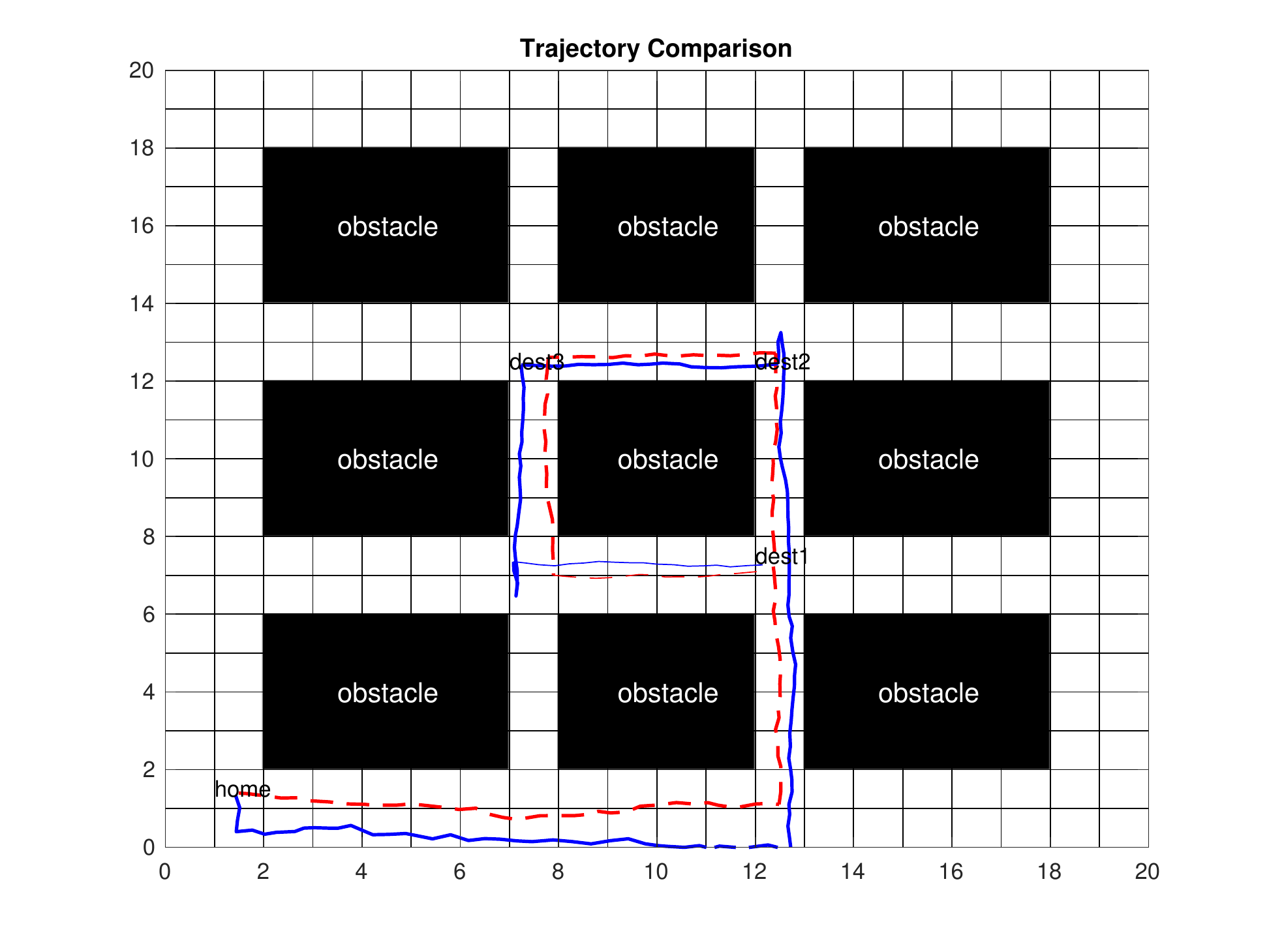}
                 \subcaption {}
                 \label{fig:traj3}
                 \end{subfigure}\hfill
                 \begin{subfigure}{.43\textwidth}
                 \includegraphics[width=\textwidth]{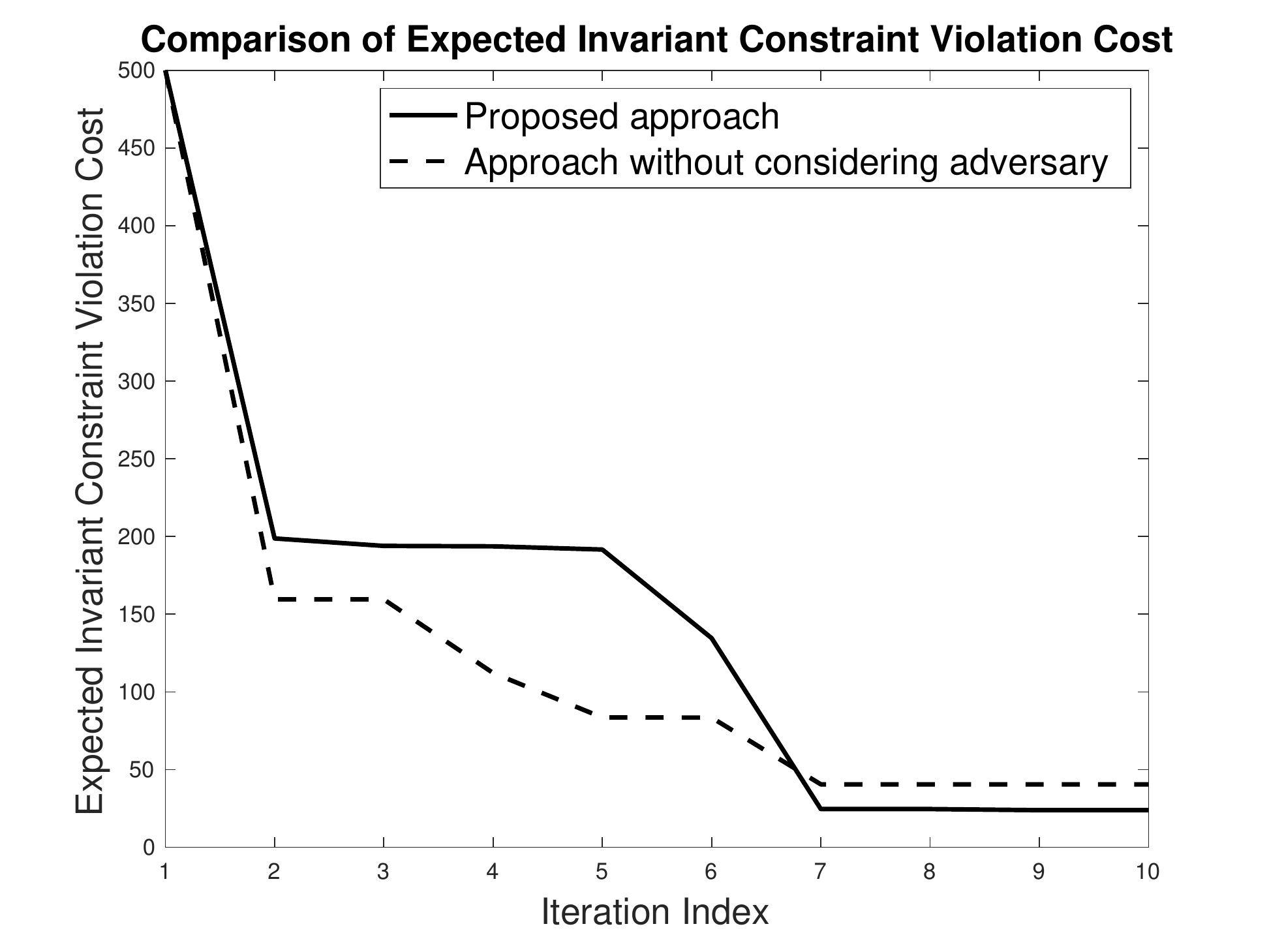}
                 \subcaption {}
                 \label{fig:cost}
                 \end{subfigure}\hfill
\caption{Comparison of the proposed approach and the approach without considering the presence of the adversary. Fig. \ref{fig:traj} gives the trajectories obtained using two approaches. The solid blue line is the trajectory obtained using the proposed approach, while the dashed red line represents the trajectory obtained using the approach without considering the presence of the adversary. Fig. \ref{fig:cost} shows the expected invariant constraint violation cost with respect to iteration indices.}
\end{figure*}

In this section, we present two case studies to demonstrate our proposed method. In particular, we focus on the application of remotely controlled UAV under deception attack. In the first case study, the UAV is given a specification modeling reach-avoid requirement. In the second case study, the UAV is given a specification modeling surveillance and collision free requirement. Both case studies were run on a Macbook Pro with 2.6GHz Intel Core i5 CPU and 8GB RAM.

\subsection{Case Study I: Remotely Controlled UAV under Deception Attack with Reach-Avoid Specification}\label{ex:case 1}

\begin{table*}[t!]
\centering
\begin{tabular}{|c|c|c|c|c|c|c|c|c|}
  \hline
   State & $(7,8)$ & $(8,8)$ & $(13,8)$ & $(14,8)$ & $(7,13)$ & $(8,13)$ & $(13,13)$ &$(14,13)$ \\
   \hline
  $Pr^{\mu\tau}_\mathcal{G}$ & 0.6684 & 0.6028 &0.5915 & 0.4893 & 0.8981 & 0.7126 & 0.6684 & 0.6028 \\
  \hline
  $Pr^{\tilde{\mu}\tilde{\tau}}_\mathcal{G}$ & 0.3619 & 0.3182 & 0.2878 & 0.1701 & 0.6146 & 0.5112 & 0.3619 & 0.3182 \\
  \hline
  Improvement & $84.69\%$ & $89.44\%$ & $105.52\%$ &$187.65\%$ &$46.13\%$ & $39.40\%$ &$84.69\%$ & $89.44\%$\\
  \hline
\end{tabular}
\caption{Comparison of probabilities of satisfying specification $\phi$ when starting from the states located in intersections using proposed approach and approach without considering the adversary.}
\label{table:prob}
\end{table*}

In this case study, we focus on the application of remotely controlled UAV, which conducts package delivery service. The UAV carries multiple packages and is required to deliver the packages to pre-given locations in particular order (e.g., the solution of a travelling salesman problem). The UAV navigates in a discrete bounded grid environment following system model $x(t+1) = x(t) + \left(u_{C}(t) + u_{A}(t)+\vartheta(t)\right)\Delta t,$ where $x(t)\in\mathcal{R}^2$ is the location of the UAV, $u_C(t)\in\mathcal{U}\subseteq\mathcal{R}^2$ is the control signal, $u_A(t)\in\mathcal{A}\subseteq\mathcal{R}^2$ is the attack signal, $\vartheta(t)\in\mathcal{D}\subseteq\mathcal{R}^2$ is the stochastic disturbance and $\Delta t$ is time interval. In particular, we let the control set $\mathcal{U}=[-0.3,0.3]^2$, the attack action signal set $\mathcal{A}=[-0.2,0.2]^2$, the disturbance set $\mathcal{D}=[-0.05,0.05]^2$. Also, the disturbance $\vartheta(t)\sim\mathcal{N}(0,\Gamma)$, where $\Gamma=\text{diag}(0.15^2,0.15^2)$.

We abstract the system as an SG using Algorithm \ref{alg: CMDP}. In particular, given the location of the UAV, we can map the location of the UAV to the grid and simulate the grid it reaches at time $t+1$. Each grid in the environment can be mapped to a state in the SG. In this case study, there exists $400$ states in the SG. In the following, we use location and state interchangeably. The control actions and attack signals are sets of discrete control inputs. The label of each state is shown in Fig. \ref{fig:traj}. The transition probability can be obtained using Algorithm \ref{alg: CMDP}. 

The UAV is required to deliver packages to three locations `dest1', `dest2', and `dest3' in this particular order after departing from its `home'. Then it has to return to `home' and stay there forever. Also during this delivery service, the UAV should avoid colliding with obstacle areas marked as black in Fig. \ref{fig:traj} to Fig. \ref{fig:traj2}. The LTL formula is written as $\phi=\text{home}\land\Diamond(\text{dest}1\land\Diamond(\text{dest}2\land\Diamond\text{dest}3))\land\Diamond\Box\text{home}\land\Box\neg\text{obstacle}.$ 

We compare the control policy obtained using the proposed approach with that synthesized without considering the presence of the adversary. In Fig. \ref{fig:traj}, we present the sample trajectories obtained using these approaches. The solid line shows a sample trajectory obtained by using the proposed approach, and the dashed line shows the trajectory obtained by using the control policy synthesized without considering the presence of adversary. To demonstrate the resilience of the proposed approach, we let the states located in the intersections be labelled as `home' and hence are set as the initial states. We compare the probability of satisfying the specification $\phi$ using the proposed approach and the approach without considering the adversary in Fig. \ref{fig:traj1} and Fig. \ref{fig:traj2}, respectively. We observe that the control policy synthesized using the proposed approach has higher probability of satisfying specification $\phi$. The detailed probability of satisfying specification $\phi$ is listed in Table \ref{table:prob}. Denote the probability of satisfying specification $\phi$ using the proposed approach and the approach without considering the adversary as $Pr^{\mu\tau}_\mathcal{G}$ and $Pr^{\tilde{\mu}\tilde{\tau}}_\mathcal{G}$, respectively. By using the proposed approach, the average of the improvements of the probability of satisfying the given specification starting from intersection states achieves $({Pr^{\mu\tau}_\mathcal{G}-Pr^{\tilde{\mu}\tilde{\tau}}_\mathcal{G}})/{Pr^{\tilde{\mu}\tilde{\tau}}_\mathcal{G}} = 90.87\%.$

The computation of transition probability took $890$ seconds. Given the transition probability, the SG and DRA associated with specification $\phi$ are created within $1$ and $0.01$ second, respectively. The computation of product SG took $80$ seconds. The product SG has $2000$ states and $41700$ transitions. It took $45$ seconds to compute the control policy.

\subsection{Case Study II: Remotely Controlled UAV under Deception Attack with Liveness and Invariant Specification}

In this case study, we focus on the same UAV model as presented in Section \ref{ex:case 1}. Let the UAV be given an LTL specification $\phi=\Box(\Diamond(\text{dest}1\land\Diamond(\text{dest}2\land\Diamond\text{dest}3)))\land\Box\neg\text{obstacle}$ consisting of liveness and invariant constraints. In particular, the lieveness constraint $\phi_1=\Box(\Diamond(\text{dest}1\land\Diamond(\text{dest}2\land\Diamond\text{dest}3)))$ models a surveillance task, i.e., the UAV is required to patrol three critical regions infinitely often following a particular order, and the invariant constraint $\psi=\Box\neg\text{obstacle}$ requires the UAV to avoid collisions with obstacles. Once the critical regions are visited, a cycle is completed. During each cycle, the rate of invariant constraint violation need to be minimized. The cost incurred at each violation is assigned to be $20$. 

We compare the proposed approach with the approach without considering the adversary. The sample trajectories obtained using these approaches are presented in Fig. \ref{fig:traj3}. In particular, the solid line shows a sample trajectory obtained by using the proposed approach, and the dashed line shows the trajectory obtained by using the control policy synthesized without considering the presence of adversary. We observe that the control strategy of synthesized using the approach without considering the adversary uses less effort comparing to the proposed approach. However, the proposed approach is more resilient since it uses more control effort to deviate from the obstacles to minimize the violation cost. We present the average invariant constraint violation cost incurred using the control policy obtained at each iteration in Fig. \ref{fig:cost}. We observe that the proposed approach incurs lower cost after convergence. In Fig. \ref{fig:cost}, the approach that does not consider the adversary incurs lower cost compared to the proposed approach during iterations $2$ to $6$. The reason is that although the proposed approach guarantees convergence to Stackelberg equilibrium, it does not guarantee optimality of the intermediate policies. The average invariant constraint violation cost using proposed approach is $23.90$, while the average invariant constraint violation cost using the approach without considering the adversary is $40.52$. The improvement achieved using the proposed approach is $28.67\%$. 

Given the transition probability, the SG and DRA associated with specification $\phi$ are created within $1$ and $0.01$ second, respectively. The computation of product SG took $72$ seconds. The product SG has $1600$ states and $26688$ transitions. It took $36$ seconds to compute the control policy.
%%%%%%%%%%%%%%%%%%%%%%%%%%%%%%%%%%%%%%%%%%%%%%%%%%%%%%%%%%%%%%%
\section{Conclusion}\label{sec: conclusion}
%%%%%%%%%%%%%%%%%%%%%%%%%%%%%%%%%%%%%%%%%%%%%%%%%%%%%%%%%%%%%%%
%In this paper, we investigated the problem of satisfying LTL safety and liveness specifications in the presence of malicious adversaries. We formulated a stochastic Stackelberg game for selecting a stationary control policy that maximizes the probability of satisfying the constraints in the presence of an adversary who observes the policy and chooses a strategy to minimize the satisfaction probability. We mapped the problem to an equivalent reachability game, and proposed a deterministic polynomial-time algorithm for computing an optimal control strategy. Future work will consider more general LTL specifications  and non-stationary control and adversary strategies.

In this paper, we investigated two problems on a discrete-time dynamical system in the presence of an adversary. We assumed that the adversary can initiate malicious attacks on the system by observing the control policy of the controller and choosing an intelligent strategy. First, we studied the problem of maximizing the probability of satisfying given LTL specifications. A stochastic Stackelberg game was formulated to compute a stationary control policy. A deterministic polynomial-time algorithm was proposed to solve the game. Second, we formulated the problem of minimizing the expected times of invariant constraint violation while maximizing the probability of satisfying the liveness specification. We developed a policy iteration algorithm to compute an optimal control policy by exploiting connections to the ACPS problem.  The bottleneck of the proposed framework is the computation complexity of the abstraction procedure. However, this is beyond the scope of this work. The potential ways to reduce the computation complexity include exploring the symmetry of the environment, and applying receding horizon based control framework. In future work, we will consider non-stationary control and adversary policies.

%In this paper, two problems have been investigated for a dynamical system with the presence of strategic adversary. First, we studied how to compute an optimal control policy for the controller to maximize the probability of satisfying the given specification modeled in an LTL formula. Second, we considered the long term performance of the system with respect to the ACPC problem. We developed the optimality conditions for ACPC problem, and presented an algorithm using value iteration to create the control policy. A case study is presented to demonstrate our theoretical analysis.

\bibliographystyle{IEEEtran}
\bibliography{TAC2017}

\vskip -1.8\baselineskip plus -1fil
\begin{IEEEbiography}[{\includegraphics[width=1in,height=1.25in,clip,keepaspectratio]{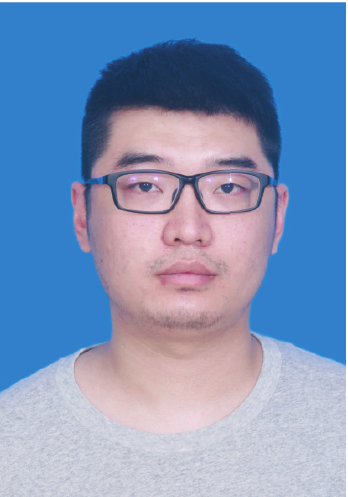}}]
{Luyao Niu}(SM'15)
received the B.Eng. degree from the School of Electro-Mechanical Engineering, Xidian University, Xi’an, China, in 2013 and the M.Sc. degree from the Department of Electrical and Computer Engineering, Worcester Polytechnic Institute (WPI) in 2015. He has been working towards his Ph.D. degree in the Department of Electrical and Computer Engineering at Worcester Polytechnic Institute since 2016. His current research interests include optimization, game theory, and control and security of cyber physical systems. %He is a student member of the IEEE.
\end{IEEEbiography}
\vskip -1.8\baselineskip plus -1fil

\begin{IEEEbiography}[{\includegraphics[width=1in,height=1.25in,clip,keepaspectratio]{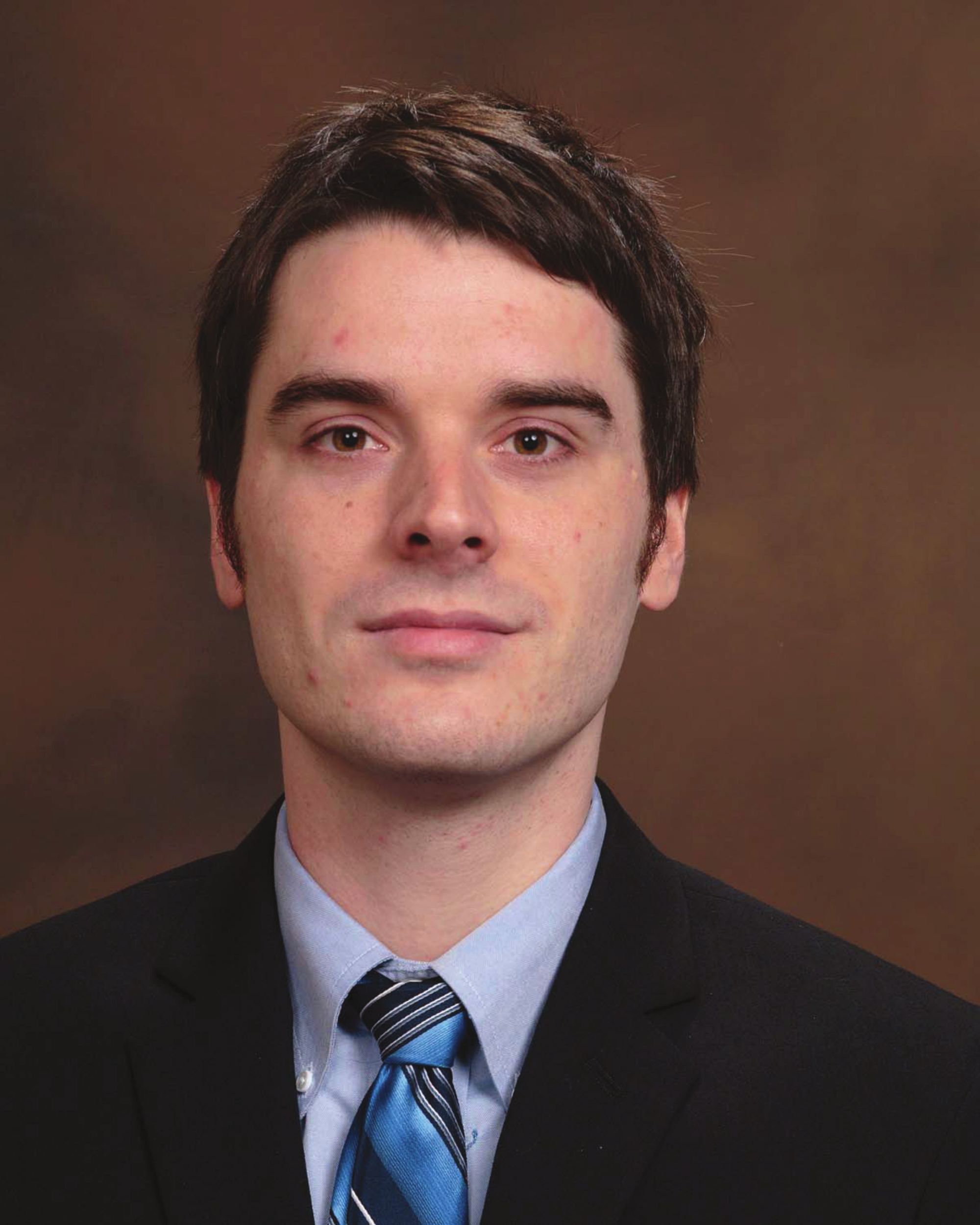}}]{Andrew Clark}(M'15)
is an Assistant Professor in the Department of Electrical and Computer Engineering at Worcester Polytechnic Institute. He received the B.S. degree in Electrical Engineering and the M.S. degree in Mathematics from the University of Michigan - Ann Arbor in 2007 and 2008, respectively. He received the Ph.D. degree from the Network Security Lab (NSL), Department of Electrical Engineering, at the University of Washington - Seattle in 2014. He is author or co-author of the IEEE/IFIP William C. Carter award-winning paper (2010), the WiOpt Best Paper (2012), and the WiOpt Student Best Paper (2014), and was a finalist for the IEEE CDC 2012 Best Student-Paper Award. He received the University of Washington Center for Information Assurance and Cybersecurity (CIAC) Distinguished Research Award (2012) and Distinguished Dissertation Award (2014). His research interests include control and security of complex networks, submodular optimization, and control-theoretic modeling of network security threats. %, and deception-based network defense mechanisms.
\end{IEEEbiography}

\end{document}